\documentclass[11pt]{article}
\usepackage[margin=1.1in]{geometry}
\usepackage{amsmath,amssymb,amsthm,bm}
\usepackage{mathtools}
\usepackage{graphicx}
\usepackage{microtype}
\usepackage[hidelinks]{hyperref}
\usepackage[round,authoryear]{natbib}
\usepackage{enumitem}
\usepackage{setspace}
\theoremstyle{plain}
\newtheorem{theorem}{Theorem}[section]
\newtheorem{proposition}[theorem]{Proposition}
\newtheorem{corollary}[theorem]{Corollary}

\theoremstyle{definition}
\newtheorem{definition}[theorem]{Definition}
\newtheorem{assumption}{Assumption}
\newtheorem{observation}[theorem]{Observation}
\theoremstyle{remark}

\makeatletter
\renewcommand\part{\@startsection{part}{0}{\z@}{0pt}{-1sp}{\normalfont\normalsize}*}
\renewcommand\thepart{\@Roman\c@part}
\let\@oldpart\part
\renewcommand\part[1]{\refstepcounter{part}\ignorespaces}
\makeatother

\newcommand{\Bten}{\mathbf{B}}
\newcommand{\Amat}{\mathbf{A}}
\newcommand{\Bp}{\mathbf{B}_p}
\newcommand{\Wmat}{\mathbf{W}}
\newcommand{\bvec}{\mathbf{b}}
\newcommand{\xvec}{\mathbf{x}}
\newcommand{\Mmat}{\mathbf{M}}
\newcommand{\Emat}{\mathbf{E}}
\newcommand{\Pmat}{\mathbf{P}}
\newcommand{\Cmat}{\mathbf{C}}
\newcommand{\hvec}{\mathbf{h}}
\newcommand{\cvec}{\mathbf{c}}
\newcommand{\Top}{\mathbf{T}}
\newcommand{\Kop}{\mathbf{K}}

\newcommand{\Lop}{\mathbf{L}}
\newcommand{\Dmat}{\mathbf{D}}
\newcommand{\Imat}{\mathbf{I}}
\newcommand{\one}{\mathbf{1}}
\newcommand{\bias}{\bm{\beta}}
\newcommand{\real}{\mathbb{R}}
\newcommand{\spr}{\rho}
\newcommand{\Kreiss}{\mathcal{K}}
\DeclareMathOperator{\tr}{tr}
\DeclareMathOperator{\diag}{diag}
\DeclareMathOperator{\vecop}{vec}

\allowdisplaybreaks

\title{\textbf{Epistemic Networks, Collective Misperception,\\ and the Manipulation of Social Knowledge}}
\author{Mihnea C. Moldoveanu \quad\ Joel A. C. Baum\\[3pt]
\normalsize Rotman School of Management, University of Toronto\\
\normalsize Toronto, Ontario, Canada}
\date{}

\begin{document}
\maketitle

\begin{abstract}
\noindent We investigate the structure of interactive belief: the epistemic state in which agents hold, revise, and act on their models of one another's minds. What a group believes is not what its members believe, nor the average of what they believe. It depends on what each member takes the others to believe, and on what each takes the others to believe about still others. Ordinary models of opinion and social influence do not represent this recursion. We posit that the proper unit of social-epistemic analysis is not the individual belief but the tensor of mutual attribution, the array that records, for each proposition, what every agent takes every agent to believe. We separate three layers of this object: what is privately held, what is publicly expressed, and what is perceived in others. Belief evolves by contraction of the tensor against a signed matrix of epistemic influence.

From this one construction we derive five results. First, collective misperception decomposes exactly into a component that observation dissolves and a component that observation cannot touch; pluralistic ignorance is the second component, and social conformity amplifies it without generating it. Second, the stable forms of collective epistemic life---consensus, polarization, entrenched misperception, and instability---are spectral regimes of a single operator. Third, higher-order belief reduces to walks in the epistemic network, under a stated assumption of cognitive consistency whose boundary we mark precisely and whose conclusions we prove invariant to it. Fourth, the manipulation of social knowledge is organized by the epistemic layer an attack targets; this yields a computable ranking of manipulative strategies and a sharp division between distortions that heal and distortions that sediment into architecture. Fifth, hierarchical epistemic structures are reactive: they amplify attacks that the spectral analysis certifies as safe. We develop each result alongside a worked, hand-computable example, and we interpret each in the vocabulary of social epistemology. The framework is offered not as the last word on interactive belief but as evidence that its structure is tractable.

\end{abstract}

\part{Foundations}\label{part:foundations}

\section{Introduction: the epistemology of interactive belief}\label{sec:intro}

\subsection{What a group knows}

Consider a seminar room. Every student has failed to understand the argument just presented, and every student, glancing at the composed faces nearby, concludes that they alone are lost. No one asks a question. Each reads the silence as evidence that everyone else has understood, and the class proceeds on a misunderstanding that not one of its members privately holds. Nothing false has been asserted. No one has lied. Each inference from the available evidence was locally reasonable. And yet what ``the class'' believes about its own understanding is not merely mistaken; it is inverted relative to the private beliefs of its members. This is pluralistic ignorance \citep{prentice1993,bicchieri2005}, and it is only the most familiar member of a family of phenomena in which the epistemic state of a group comes apart from any aggregate of the epistemic states of its members.

The example resists the two standard reductions, and this is why it is instructive. One reduction says that what the class believes is the sum, the majority, or a weighted average of what its members believe; collective belief is then an aggregation problem, and the tools are those of judgment aggregation and social choice \citep{list2011}. But the belief the class acts on, by proceeding, is held by no one, and it is not the average of anything the members privately think. It is a product of what each takes the others to think. The other reduction says that the group is a believer in its own right, a plural subject with intentional states not reducible to those of its members \citep{gilbert1989,tuomela2013,bratman2014}. That view is right that group belief is not mere summation, but it does not say \emph{how} the group-level state is constituted from the epistemic relations among members. The seminar's false consensus is not a brute emergent fact. It has a mechanism, and the mechanism runs through each student's model of the others.

The object we need therefore sits between the individual and the aggregate. It is the structure of \emph{mutual attribution}: what each member believes, what each expresses, and---decisively---what each takes each other member to believe. We will argue that collective epistemic states are configurations of this structure, and that their characteristic pathologies are configurations the structure makes stable. To study collective belief is to study the dynamics of interactive belief.

\subsection{Interactive belief and its recursion}

By \emph{interactive belief} we mean belief whose content includes the beliefs of others: $i$'s belief about $j$'s belief, $i$'s belief about $j$'s belief about $k$'s belief, and so on up the hierarchy. The centrality of this recursion to social life is not a new observation. It is the engine of \citet{lewis1969}'s analysis of convention, in which a regularity is a convention only if it is common knowledge that it is followed; of \citet{schelling1960}'s focal points, which coordinate expectations about expectations; of \citet{aumann1976}'s formalization of common knowledge and the agreement theorem that flows from it; and of the entire apparatus of epistemic game theory \citep{harsanyi1967,mertens1985,brandenburger2014}, in which a player's rational action depends on a hierarchy of beliefs about others' beliefs that in principle ascends without end.

Philosophically, the recursion raises a foundational question that this paper must confront directly: what is the \emph{unit} of interactive epistemic analysis, and how deep does the hierarchy go? The interactive-epistemology tradition answers, in effect, ``all the way'': a type, in the sense of \citet{harsanyi1967}, encodes an infinite hierarchy of beliefs over beliefs, and the universal type space of \citet{mertens1985} is the completed object of that ascent. This is the correct answer for the foundations of rational choice under incomplete information, but it purchases generality at the price of tractability, and---more importantly for our purposes---it does not correspond to the finite, structured, and boundedly recursive way in which situated social agents actually represent one another. Real agents do not carry infinite belief hierarchies; they carry finite models, updated from observation, that let them get on with coordinating, trusting, deferring, and dissenting. The epistemic-network or ``epinet'' program \citep{mb2011,mb2014} was built precisely to study these finite structures of mutual attribution as the empirical unit of social-epistemic analysis. The present paper gives that program its dynamical and adversarial completion.

Our strategy is to represent the interactive-belief state not by an infinite hierarchy but by a finite array from which higher-order attributions are \emph{generated by composition}, under an explicitly stated assumption of cognitive consistency (Section~\ref{sec:higher}). This is a substantive modeling commitment, and we will be careful to mark exactly what it captures---the finite arithmetic of who-takes-whom-to-believe-what---and what it sets aside---the full modal structure of interactive knowledge, including the distinction between mutual and common knowledge on which some of the deepest coordination results turn. Within its stated scope, the payoff is that the entire recursion becomes linear algebra, and the machinery of spectra, invariant subspaces, and resolvents can be brought to bear on questions that have previously been treated only informally or only in the infinitary setting.

\subsection{Three literatures, one structure}

Three largely separate literatures have independently reached the same conclusion: collective outcomes are driven less by what people believe than by what they believe others believe. Their convergence is part of what motivates a unified treatment.

The first is the economics of speculative bubbles. In the modern theory, an asset is held not because the holder values it at the price but because she believes others believe still others will buy it later; the bubble is sustained by a tower of beliefs about beliefs, and can persist even when it is common knowledge that the asset is overpriced \citep{simsek2021,awaya2022,liao2022}. What circulates through such a market is not a distribution of private valuations but a structure of higher-order attributions. The greater fool is a second-order epistemic object.

The second is the physics and network science of opinion dynamics---the study of how consensus, polarization, and echo chambers emerge from local rules of social influence \citep{degroot1974,friedkin1990,hegselmann2002,acemoglu2011,baumann2020,wang2020,leonard2021}. This literature has produced a rich dynamical vocabulary, but it almost universally assigns each agent a single opinion and folds the perception of others' opinions into the update rule. The represented state is first-order; the second-order structure---what agents take others to think, as distinct from what others think---is not carried as a variable that can evolve, lag, or be manipulated on its own.

The third is the social psychology and economics of misperception: pluralistic ignorance, preference falsification, and the systematic gap between private attitudes and their public expression \citep{prentice1993,kuran1995,bicchieri2005,bursztyn2022,sparkman2022}. Here the divergence between what is believed and what is believed to be believed is the central object, documented at scale---a majority misreading the majority---but the models tend to be issue-specific, static, or binary, and they are not connected to the higher-order recursion that the bubble literature shows to be decisive, nor cast as tractable dynamical systems.

Each literature holds one part of the structure and leaves another dark. Markets supply the higher-order recursion but no dynamics of misperception. Opinion dynamics supplies the dynamics but no separation of private from perceived belief. The misperception literature supplies that separation but neither the recursion nor a tractable operator theory. We design the framework of this paper to hold all three at once, in one object, with one analytical engine. We turn to constructing it.

\subsection{Plan of the paper}

Part~\ref{part:foundations} builds the framework: the attribution tensor and the three-layer ontology of belief (Section~\ref{sec:tensor}), the contraction dynamics and the exact decomposition of collective misperception (Section~\ref{sec:contraction}), and the treatment of higher-order belief as composition, with its philosophical boundary made explicit (Section~\ref{sec:higher}). Part~\ref{part:regimes} develops the four stable forms of collective epistemic life as spectral regimes of one operator, each with a fully worked example, culminating in the closure theorem that is the formal heart of the paper and the ``amplifies-but-cannot-generate'' law for pluralistic ignorance (Sections~\ref{sec:consensus}--\ref{sec:synthesis}). Part~\ref{part:coupling} activates the proposition dimension of the tensor, giving cross-proposition inference and its destabilizing power. Part~\ref{part:adversarial} turns the theory to the manipulation of social knowledge, developing the layer-targeted taxonomy of epistemic attack, the cost-of-persuasion crossover, the sedimentation of distortion into architecture, and the perversity of misdirected transparency. Part~\ref{part:physics} establishes the transient fragility of hierarchical epistemic structures. Part~\ref{part:synthesis} draws the philosophical implications together and marks the limits of the linear theory. Full proofs are given in the main text throughout; only the four classical regime results, which we restate rather than originate, are proved in an appendix.

\section{The attribution tensor and the three layers of belief}\label{sec:tensor}

\subsection{Why a tensor}

The claim that the unit of social-epistemic analysis is the structure of mutual attribution can be made precise. Fix a finite population of agents $A=\{1,\dots,N\}$ and a finite set of propositions $Q=\{1,\dots,K\}$ at issue among them. A first-order description assigns to each agent and proposition a degree of belief: an $N\times K$ array. This is the object of opinion dynamics, and it cannot express the interactive content we have argued is essential, because it has no place to record what agent $i$ takes agent $j$ to believe. To record that, we need a third index. The minimal object that simultaneously specifies \emph{who is doing the attributing}, \emph{who is being modeled}, and \emph{which proposition is at issue} is a third-order tensor:
\begin{equation}
\Bten \in \real^{N\times N\times K}, \qquad B_{ijp} = \text{agent } i\text{'s estimate of agent } j\text{'s degree of belief in proposition } p.
\end{equation}
The three indices are not interchangeable. The first is the \emph{observer} or attributing agent; the second is the \emph{target} being modeled; the third is the proposition. That the first two indices range over the same set $A$ is what makes the object genuinely reflexive: agents model agents, including themselves, and the diagonal in the first two indices will turn out to carry the first-order beliefs as a special case.

It is worth pausing on the philosophical force of this representational choice. In treating $B_{ijp}$ as the primitive, we are committing to the view that the fundamental social-epistemic fact is not ``$j$ believes $p$'' but ``$i$ takes $j$ to believe $p$''---an attribution, indexed by an attributor. The first-order fact ``$j$ believes $p$'' is recovered as the self-attribution $B_{jjp}$, the case in which the attributor and target coincide. This is not a denial that there are first-order beliefs; it is a claim about their place in the social-epistemic order. What moves through a group, what agents act on, what can be manipulated, and what can come apart from reality, is the attribution, not the belief simpliciter. The tensor makes the attribution the citizen of first class and the belief a derived, diagonal special case---an inversion of the usual priority that we take to be the correct one for social epistemology.

Fixing a proposition $p$ collapses the tensor to a single \emph{slice}, an $N\times N$ matrix
\begin{equation}
\Bp = [B_{ijp}] \in \real^{N\times N},
\end{equation}
which records, for that one proposition, the entire structure of who takes whom to believe it, and to what degree. Almost all of the dynamics act one slice at a time, so for most of the paper we fix a proposition and suppress the index $p$, restoring it only in Part~\ref{part:coupling}, where inference \emph{across} propositions is exactly what makes the third mode of the tensor come alive. Until then, the working object is a single square matrix $\Bten = [B_{ij}]$ of mutual attributions on a fixed proposition, and the reader may keep the seminar room in mind: $B_{ij}$ is how confident student $i$ takes student $j$ to be that the argument was understood.

A methodological clarification is in order here, since the tensor framing could otherwise promise more than the mathematics delivers. The third-order tensor is the right object for the \emph{ontological} claim of this paper---that the unit of social-epistemic analysis carries three indices, observer, target, and proposition---and it is genuine tensorial bookkeeping. But we do not want to overstate its \emph{mathematical} role. For the whole of Parts~\ref{part:regimes}, and until inferential coupling is switched on in Part~\ref{part:coupling}, the proposition mode is inert: the theory is proposition-separable, the tensor is a stack of independent $N\times N$ slices, and every result is a theorem of linear-operator theory on a single slice. The genuinely multilinear content is confined to Part~\ref{part:coupling}, and even there it reduces, on vectorization, to a Kronecker product $\Lop\otimes\Wmat$---that is, back to matrix analysis, via the product-spectrum law. We therefore make no claim to novel multilinear-algebraic machinery; the contribution is the identification of the three-index attribution tensor as the correct state object for interactive belief, together with the operator theory that governs its evolution. Where the mode-3 structure does real work (Part~\ref{part:coupling}), we exploit it explicitly; elsewhere, ``tensor'' should be read as the natural indexing of a family of attribution matrices, not as a promise of higher multilinearity. A fuller exploitation of the tensor structure---for instance, low-rank (Tucker or canonical-polyadic) factorizations of $\Bten$ as a model of \emph{bounded} theory of mind, with the rank as a cognitive-capacity parameter---is a natural direction the present paper leaves open, and we return to it in the concluding discussion.

\subsection{Three layers: private, expressed, perceived}

The tensor's diagonal and off-diagonal entries play distinct epistemic roles, and drawing the distinction carefully yields a three-layer ontology of belief that is the conceptual core of the framework.

\begin{definition}[Private belief]\label{def:private}
The \emph{private belief} of agent $i$ on the fixed proposition is the diagonal entry $b_i := B_{ii}$, agent $i$'s own degree of belief---what $i$ actually holds.
\end{definition}

The diagonal of the slice is thus the vector $\bvec=(b_1,\dots,b_N)$ of private beliefs, the object a first-order model would track in isolation. The off-diagonal entries $B_{ij}$ with $i\neq j$ are $i$'s models of the others: the theory-of-mind content, the representations of other minds that situate each agent in a perceived social-epistemic environment.

But private belief and its attribution by others do not exhaust the layers we need. Between what an agent believes and what others perceive of it stands what the agent \emph{expresses}---and expression need not match conviction. This is the phenomenon of preference falsification \citep{kuran1995}, of the public/private divergence that pluralistic ignorance requires, and it demands its own state variable.

\begin{definition}[Expressed belief and the concealment gap]\label{def:expressed}
The \emph{expressed belief} $x_i$ of agent $i$ is what $i$ publicly professes on the proposition. The \emph{concealment gap} is the difference
\begin{equation}
c_i := x_i - b_i
\end{equation}
between what $i$ expresses and what $i$ privately holds. A population in which $c_i\neq 0$ for some $i$ is one in which expression has come apart from conviction.
\end{definition}

We thus have three distinct layers, and it is essential to the entire theory that they are kept apart:
\begin{enumerate}[label=(\roman*),leftmargin=2.2em,itemsep=1pt]
\item \emph{Private belief} $b_i$: what $i$ holds. The diagonal of the tensor.
\item \emph{Expressed belief} $x_i$: what $i$ professes. A separate state, related to $b_i$ by the concealment gap $c_i$.
\item \emph{Perceived belief} $B_{ij}$: what $i$ takes $j$ to hold. The off-diagonal tensor entries.
\end{enumerate}
The reduction that collapses these three into one---identifying belief, expression, and perception---is exactly the move that makes collective misperception invisible, because it assumes away the gaps ($x_i - b_i$ and $B_{ij}-b_j$) in which misperception lives. Keeping them distinct is what allows the theory to represent, and then to explain, the phenomena with which we began.

\subsection{Misperception as a first-class object}

With the layers separated, the various ways in which social belief can be in error become distinct, nameable quantities rather than undifferentiated noise. The perceptual error---the gap between what $i$ takes $j$ to believe and what $j$ actually believes---is:
\begin{definition}[Epistemic error]\label{def:error}
The \emph{epistemic error} is $E_{ij}:=B_{ij}-b_j$, the discrepancy between $i$'s model of $j$ and $j$'s private belief. The error is zero exactly when $i$ perceives $j$ veridically.
\end{definition}

This decomposition of error will prove central, because the two ways it can be nonzero correspond to two entirely different epistemic pathologies with different dynamics and different remedies. The error $E_{ij}=B_{ij}-b_j$ can be rewritten, by inserting $j$'s expression, as
\begin{equation}
E_{ij} = \underbrace{(B_{ij}-x_j)}_{\text{perception lag}} + \underbrace{(x_j - b_j)}_{\text{$j$'s concealment}} = (B_{ij}-x_j) + c_j.
\label{eq:error-split}
\end{equation}
The first term is $i$'s failure to perceive $j$'s expression accurately---a lag or noise in observation, the kind of error that more or better observation corrects. The second is $j$'s own concealment---the gap between what $j$ says and what $j$ holds, which no amount of accurate observation of $j$'s expression can penetrate, because the expression itself is what conceals. That these two sources of collective misperception are formally separable, and behave completely differently under the dynamics, is one of the framework's first substantive yields; we prove it precisely in the next section. For now the point is conceptual: once belief, expression, and perception are distinguished, misperception ceases to be a single undifferentiated failure and becomes a structured object with parts that can be told apart and treated separately.

\subsection{A note on degrees of belief}

We take the entries $B_{ijp}$ to be real numbers, representing degrees of belief on a common scale (which may be centered so that zero denotes a neutral or reference credence and the sign encodes the direction of opinion). This is the standard idealization of the opinion-dynamics literature, and it should be understood as such: a tractable proxy for a graded doxastic attitude, not a claim that credences are literally real-valued or that interpersonal comparison is unproblematic. Nothing in the qualitative theory---the decomposition of error, the classification of regimes, the manipulation results---depends on the cardinal structure in an essential way; what matters is that beliefs, expressions, and perceptions can each be higher or lower, can differ from one another, and can move. Where the cardinal scale does matter, as in the unbounded amplification results of Sections~\ref{sec:synthesis} and~\ref{sec:persuader}, we flag explicitly that the linear idealization is doing work and that a bounded-belief treatment would replace divergence with saturation. With this understood, we proceed to the dynamics.

\section{The dynamics of interactive belief}\label{sec:contraction}

\subsection{Epistemic influence and the belief update}

Agents revise their beliefs in light of what they take others to believe. To model this we need a structure of \emph{epistemic influence}: for each ordered pair $(i,j)$, a weight $W_{ij}$ recording how much agent $i$'s revision responds to its model of agent $j$. We collect these into the \emph{agreement matrix} $\Wmat\in\real^{N\times N}$. A positive weight $W_{ij}>0$ means $i$ moves toward the position it attributes to $j$---deference, trust, homophilous influence. A negative weight $W_{ij}<0$ means $i$ moves \emph{away} from that position---antagonism, distrust, the ``I believe the opposite of whatever they think'' relation that, as we will see, is the seed of polarization. That epistemic influence can be signed is essential and is one of the ways this framework departs from the stochastic-matrix orthodoxy of classical opinion dynamics.

The belief update is then the natural one: each agent sets its next private belief to the influence-weighted aggregate of the beliefs it attributes to others,
\begin{equation}
b_i^{t+1} = \sum_{j=1}^{N} W_{ij}\, B_{ij}^{t}.
\label{eq:update}
\end{equation}
Read carefully, \eqref{eq:update} already embodies the paper's central commitment. The agent updates on $B_{ij}^t$---its \emph{model} of $j$---not on $b_j^t$, $j$'s actual belief. What propagates through the network is attribution, not conviction. In the special case where all perception is veridical ($B_{ij}=b_j$ for all $i,j$), the update collapses to the classical $b_i^{t+1}=\sum_j W_{ij}b_j^t$ of \citet{degroot1974}, and the entire apparatus of first-order opinion dynamics is recovered. The framework thus contains the classical theory as its veridical-perception special case, and everything distinctive about it lives in the departure from that case.

To keep beliefs bounded we impose an absolute row-sum condition on influence:
\begin{equation}
\sum_{j=1}^{N} \lvert W_{ij}\rvert \le 1 \qquad \text{for all } i,
\label{eq:rowsum}
\end{equation}
which guarantees $\lVert\Wmat\rVert_\infty\le 1$ and hence $\spr(\Wmat)\le 1$, where $\spr$ denotes spectral radius. We note in passing that the weaker signed condition $\sum_j W_{ij}\le 1$, which suffices in the nonnegative case, does \emph{not} bound the operator norm once weights may be negative, since $\lVert\Wmat\rVert_\infty=\max_i\sum_j\lvert W_{ij}\rvert$ involves absolute values; the correct hypothesis is~\eqref{eq:rowsum}, and it contains the signed structure used for polarization as a special case.

\subsection{The truth--misperception decomposition}

The first theorem of the framework makes explicit what \eqref{eq:update} is really doing. It separates the update into a part that transmits true belief through the influence network and a part driven entirely by misperception.

\begin{theorem}[Truth--misperception decomposition]\label{thm:decomp}
The belief update \eqref{eq:update} is the affine dynamical system
\begin{equation}
\bvec^{t+1} = \Wmat\,\bvec^{t} + \hvec^{t}, \qquad \hvec^{t} = (\Wmat\odot\Emat^{t})\,\one,
\label{eq:affine}
\end{equation}
where $\odot$ is the Hadamard (entrywise) product, $\one$ is the all-ones vector, and $\Emat^t=[E_{ij}^t]$ is the epistemic-error matrix of Definition~\ref{def:error}. The first term $\Wmat\bvec^t$ transmits agents' true private beliefs through the influence network; the second term $\hvec^t$ is a forcing generated entirely by misperception, and it vanishes identically when perception is veridical.
\end{theorem}

\begin{proof}
Substitute the definition $B_{ij}^{t}=b_j^{t}+E_{ij}^{t}$ of the epistemic error into the update~\eqref{eq:update}:
\[
b_i^{t+1} = \sum_{j} W_{ij}\bigl(b_j^{t}+E_{ij}^{t}\bigr) = \sum_{j} W_{ij}b_j^{t} + \sum_{j} W_{ij}E_{ij}^{t}.
\]
The first sum is the $i$-th component of the matrix--vector product $\Wmat\bvec^{t}$. The second is $\sum_j W_{ij}E_{ij}^t$, which is precisely the $i$-th component of $(\Wmat\odot\Emat^{t})\one$, since the Hadamard product $\Wmat\odot\Emat^t$ has $(i,j)$ entry $W_{ij}E_{ij}^t$ and right-multiplication by $\one$ sums each row. If perception is veridical then $E_{ij}^t=0$ for all $i,j$, so $\Emat^t=\mathbf 0$ and hence $\hvec^t=\mathbf 0$.
\end{proof}

The decomposition is worth dwelling on. It says that a population of interactive believers evolves \emph{as if} it were a classical DeGroot population ($\Wmat\bvec^t$) perturbed by an external forcing ($\hvec^t$)---but the forcing is not external at all. It is generated endogenously, at every instant, by the population's own collective misperception, through the Hadamard contraction of the influence matrix against the error matrix. Collective misbelief is thus not noise added to an otherwise-clean process; it is a structured driving term with an exact expression, $(\Wmat\odot\Emat^t)\one$, whose anatomy we can now dissect. The whole of the manipulation theory in Part~\ref{part:adversarial} will consist in asking who can write into this forcing term, through which layer, and at what cost.

\subsection{Closing the system: expression and perception}

Theorem~\ref{thm:decomp} evolves only the private beliefs $\bvec$, yet its forcing term depends on the full error matrix, hence on the off-diagonal perceptions $B_{ij}$ and, through them, on expressions $x_j$. The system is therefore not yet closed: we have a law of motion for private belief but not for expression or perception. Closing it requires saying how agents decide what to express, and how they update what they perceive. These two maps are where the social psychology enters, and they are what make misperception \emph{generated} rather than assumed.

\paragraph{The expression map.} What an agent professes is a compromise between conviction and the perceived climate of opinion. We model expressed belief as a convex combination of the agent's own updated private belief and the average position it perceives around it, plus a possible fixed bias:
\begin{equation}
x_i^{t+1} = (1-\alpha_i)\,b_i^{t+1} + \alpha_i\,\widehat{c}_i^{\,t} + \beta_i, \qquad
\widehat{c}_i^{\,t} = \frac{\sum_{j\neq i}\lvert W_{ij}\rvert\,B_{ij}^{t}}{\sum_{j\neq i}\lvert W_{ij}\rvert}.
\label{eq:expression}
\end{equation}
Here $\alpha_i\in[0,1]$ is agent $i$'s \emph{conformity pressure}: the weight it places on saying what it perceives others to be saying, as opposed to what it privately believes. The quantity $\widehat c_i^{\,t}$ is the \emph{perceived climate of opinion}---the influence-weighted average of the expressions $i$ attributes to others. And $\beta_i$ is a fixed \emph{expressive bias}, capturing any systematic pull on expression that is independent of both conviction and climate: social desirability, an incentive to profess a particular position, the cost of voicing one view rather than another. When $\alpha_i=0$ and $\beta_i=0$ the agent is sincere, expressing exactly its private belief; as $\alpha_i\to 1$ the agent professes only the perceived consensus, whatever it privately holds; and $\beta_i$ tilts the whole thing toward a favored position. The three terms correspond to the three classical determinants of public speech: what one thinks, what one thinks is expected, and what one is rewarded for saying.

\paragraph{The perception map.} Agents come to know others' expressions by observing them over time, imperfectly. We model perception as exponential smoothing: each agent adjusts its model of another toward that other's currently expressed belief, at a rate that captures the responsiveness of observation:
\begin{equation}
B_{ij}^{t+1} = B_{ij}^{t} + \gamma_{ij}\bigl(x_j^{t} - B_{ij}^{t}\bigr), \quad i\neq j, \qquad B_{ii}^{t+1}=b_i^{t+1}.
\label{eq:perception}
\end{equation}
The \emph{perception rate} $\gamma_{ij}\in(0,1]$ governs how quickly $i$'s model of $j$ tracks $j$'s expression: $\gamma_{ij}=1$ is instantaneous, perfect observation of expression; small $\gamma_{ij}$ is slow, laggy perception. The diagonal is pinned to private belief, $B_{ii}=b_i$, by Definition~\ref{def:private}: an agent's model of itself is simply its own belief. Crucially, perception tracks \emph{expression} $x_j$, not private belief $b_j$---agents observe what others say, not what they think---and this is exactly why the concealment component of misperception, the $c_j$ in~\eqref{eq:error-split}, cannot be corrected by better perception: no matter how accurately $i$ tracks $j$'s expression, if $j$ conceals, $i$'s model of $j$'s expression converges to $j$'s expression, which already differs from $j$'s belief.

With~\eqref{eq:update},~\eqref{eq:expression}, and~\eqref{eq:perception} in hand, the triple $(\bvec^t,\xvec^t,\Bten^t)$ evolves as a closed, autonomous dynamical system. The misperception forcing $\hvec^t$ is no longer an input but an output: it is produced, at every step, by the interaction of laggy perception and strategic expression. We can now say exactly what it is made of.

\subsection{The two channels of collective misperception}\label{sec:channels}

Substituting the error decomposition~\eqref{eq:error-split} into the forcing term of Theorem~\ref{thm:decomp} splits the driving term of collective misbelief into two channels with, as promised, entirely different characters.

\begin{corollary}[Channel decomposition of the forcing]\label{cor:channels}
The misperception forcing factors additively as
\begin{equation}
\hvec^{t} = \hvec_{\mathrm{perc}}^{t} + \hvec_{\mathrm{con}}^{t}, \qquad
\hvec_{\mathrm{perc}}^{t} = (\Wmat\odot\Pmat^{t})\one, \quad
\hvec_{\mathrm{con}}^{t} = (\Wmat\odot\Cmat^{t})\one,
\label{eq:channel-split}
\end{equation}
where $\Pmat^{t}=[B_{ij}^{t}-x_j^{t}]$ is the \emph{perception-lag matrix} (how far each agent's model of another trails that other's current expression) and $\Cmat^{t}=[c_j^{t}]$ is the \emph{concealment matrix} (how far each agent's expression trails its own belief), both with zero diagonal.
\end{corollary}

\begin{proof}
The Hadamard contraction $(\Wmat\odot\,\cdot\,)\one$ is linear in its second argument. By~\eqref{eq:error-split}, off the diagonal $E_{ij}^t=(B_{ij}^t-x_j^t)+c_j^t$, i.e.\ $\Emat^t=\Pmat^t+\Cmat^t$ as matrices with zero diagonal. Linearity gives $(\Wmat\odot(\Pmat^t+\Cmat^t))\one=(\Wmat\odot\Pmat^t)\one+(\Wmat\odot\Cmat^t)\one$.
\end{proof}

The two channels are the two ways a population can be collectively wrong about itself, and their difference is the difference between an error that self-corrects and an error that does not. The \emph{perception-lag channel} $\hvec_{\mathrm{perc}}$ is driven by the gap between what agents perceive and what is currently expressed; it is the error of imperfect observation. As we will prove, under any stationary pattern of expression this channel decays to zero: perception catches up to expression, and the lag closes. It is a transient. The \emph{concealment channel} $\hvec_{\mathrm{con}}$ is driven by the gap between expression and belief---by the fact that agents are not saying what they think. This channel does not close through observation, because the thing observed (expression) is itself the locus of the distortion. It persists as long as agents conceal.

We can now state, still informally, the mechanism of pluralistic ignorance that the framework delivers, and which the remainder of Part~\ref{part:regimes} makes rigorous: \emph{collective self-deception is what remains in the forcing term after the perception lag has closed}. It is the residue of the concealment channel. A group settles into pluralistic ignorance not because its members cannot see one another---they can, and they do---but because what they see, expression, has been detached from what is held, belief, by the pressure to conform and the bias in what it is safe or rewarding to say. This is why, as we will show, more observation does not cure pluralistic ignorance, and can deepen it: it closes the one channel that was already going to close, while leaving untouched, or even feeding, the channel that actually sustains the distortion. But to establish these claims we must show that the closed system converges at all, and characterize what it converges to. That is the work of Part~\ref{part:regimes}.

\section{Higher-order belief as composition, and its philosophical boundary}\label{sec:higher}

\subsection{The recursion made finite}

We have insisted that interactive belief is recursive: $i$'s belief about $j$'s belief about $k$'s belief, and onward. The attribution tensor as introduced carries only the first-order attributions $B_{ij}$---what $i$ takes $j$ to believe. Where, in this apparatus, is the higher-order content? The answer is the technical and philosophical fulcrum of the framework, and it must be handled with care, because it is here that a tractable social epistemology parts company with the infinitary hierarchies of interactive epistemology.

A naive representation of higher-order belief would add indices without end: a fourth-order tensor for ``$i$'s model of $j$'s model of $k$,'' a fifth for the next level, and so on, reproducing inside the model the very infinite ascent that makes the universal type space infinite-dimensional. This is faithful but useless; it does not correspond to how finite agents represent one another, and it forecloses analysis. The alternative we adopt is to \emph{generate} higher-order attributions by composing the first-order slice with itself. Consider the square of the attribution matrix on a fixed proposition:
\begin{equation}
(\Bten^{2})_{ij} = \sum_{m=1}^{N} B_{im} B_{mj}.
\label{eq:square}
\end{equation}
The summand $B_{im}B_{mj}$ is the product of ``$i$'s model of $m$'' and ``$m$'s belief as $i$ would carry it''---read as ``$i$'s model of $m$'s model of $j$,'' weighted and summed over the intermediary $m$. The matrix square thus \emph{computes} second-order attribution from first-order data, and the general pattern is a theorem about walks.

\begin{theorem}[Path aggregation]\label{thm:path}
For any integer $L\ge 1$, the $L$-th matrix power of the attribution slice resolves into a sum over directed walks of length $L$ in the epistemic network:
\begin{equation}
(\Bten^{L})_{ij} = \sum_{i_1,\dots,i_{L-1}} B_{i\,i_1}\, B_{i_1 i_2}\cdots B_{i_{L-1}\,j},
\label{eq:paths}
\end{equation}
each term being the product of attribution weights along the walk $i\to i_1\to\cdots\to i_{L-1}\to j$, interpreted as the nested attribution ``$i$'s model of $i_1$'s model of $\cdots$ of $j$'s belief.''
\end{theorem}

\begin{proof}
By induction on $L$. For $L=1$ the claim is the definition of the slice, $(\Bten^1)_{ij}=B_{ij}$. Assume~\eqref{eq:paths} holds for $L$. Then, by the definition of matrix multiplication,
\[
(\Bten^{L+1})_{ij} = \sum_{k} (\Bten^{L})_{ik}\, B_{kj} = \sum_k \Bigl(\sum_{i_1,\dots,i_{L-1}} B_{i\,i_1}\cdots B_{i_{L-1}\,k}\Bigr) B_{kj}.
\]
Relabeling the final intermediary $k=i_L$ and absorbing it into the sum yields a sum over all walks $i\to i_1\to\cdots\to i_L\to j$ of length $L+1$, which is~\eqref{eq:paths} at $L+1$.
\end{proof}

The theorem says that the higher-order recursion, rather than requiring higher-order tensors, is \emph{already present} in the powers of the first-order slice, exactly as walks in a weighted directed graph are present in the powers of its adjacency matrix. To hold a belief about a belief about a belief is, in this framework, to traverse a length-three walk in the epistemic network, accumulating attribution weight multiplicatively along the way. Two specializations are worth naming. Applying the power to the private-belief vector, $(\Bten^L\bvec)_i$, terminates a nested attribution in an actual belief---``$i$ thinks $i_1$ thinks $\cdots$ thinks $i_L$ \emph{believes} $p$''---anchoring the chain in a ground-truth doxastic fact. And setting $j=i$ isolates \emph{closed} walks, the diagonal entries $(\Bten^L)_{ii}$, which represent reflexive higher-order belief: ``$i$'s view of how others view $i$,'' the recursive self-reference that, as we will see in Part~\ref{part:adversarial}, is the structural signature of the echo chamber---understood, per the invariance analysis of Section~\ref{sec:invariance}, as a walk-weight phenomenon (the compounding of reflexive attribution weight) rather than a claim about the modal status of what the population commonly believes.

\begin{theorem}[Closed epistemic walks and bounded recursion]\label{thm:closed}
The diagonal of the $L$-th power collects closed epistemic walks of length $L$ based at each agent, and its trace sums all such loops in the population:
\begin{equation}
(\Bten^{L})_{ii} = \sum_{i_1,\dots,i_{L-1}} B_{i\,i_1}\cdots B_{i_{L-1}\,i}, \qquad \tr(\Bten^{L}) = \sum_{i}(\Bten^{L})_{ii}.
\end{equation}
When $\spr(\Bten)<1$, these contributions decay geometrically in the depth $L$; when $\spr(\Bten)\ge 1$, deep loops can dominate.
\end{theorem}

\begin{proof}
The diagonal formula is~\eqref{eq:paths} with $j=i$, and the trace is the sum of diagonal entries by definition. For the magnitude, $\lvert(\Bten^L)_{ii}\rvert\le\lVert\Bten^L\rVert$, and for any submultiplicative norm $\lVert\Bten^L\rVert^{1/L}\to\spr(\Bten)$ by Gelfand's formula, so the entries decay geometrically at rate approaching $\spr(\Bten)$ when $\spr(\Bten)<1$ and need not decay when $\spr(\Bten)\ge 1$.
\end{proof}

\subsection{The cognitive-consistency assumption and what it sets aside}

The move from~\eqref{eq:square} onward rests on a substantive assumption that we now state explicitly, because it is the boundary between this theory and the full interactive-epistemology hierarchy, and honesty about that boundary is part of the philosophical contribution.

\begin{assumption}[Compositional attribution]\label{ass:composition}
An agent's higher-order attributions factor through its first-order models: agent $i$'s estimate of $j$'s estimate of $k$'s belief is given by the composition of $i$'s first-order attributions, i.e.\ by the matrix products of Theorem~\ref{thm:path}, rather than being carried as independent, freely specifiable higher-order data.
\end{assumption}

Assumption~\ref{ass:composition} is a strong hypothesis of cognitive consistency. It says that an agent does not separately store ``what I think $j$ thinks $k$ believes'' as a free parameter, but computes it from ``what I think $j$ thinks'' and ``what I think $k$ believes'' by composition---that the agent's representation of others is, in this specific sense, coherent across levels. Real agents surely violate this: one can believe that Anwar thinks Beatriz is confident while separately, and inconsistently, believing that Anwar thinks Beatriz thinks the matter is hopeless. The assumption idealizes such inconsistencies away, and in doing so it fixes exactly what the framework can and cannot represent.

What it captures is the finite arithmetic of nested attribution---the who-takes-whom-to-believe-what, propagated by composition, which is enough to model gossip, deference chains, reputational cascades, and the recursive structure of speculative belief. What it sets aside is the full modal content of interactive \emph{knowledge}. In the tradition of \citet{lewis1969}, \citet{aumann1976}, and the epistemic-game-theory program \citep{harsanyi1967,mertens1985,brandenburger2014}, higher-order attitudes are not scalar point-estimates composed by multiplication. They are probability measures over measures, or knowledge operators satisfying introspection axioms. And the distinctions that matter most in that tradition---above all, between $p$ being mutually believed to depth $n$ and $p$ being \emph{common} knowledge, the limit that coordinates conventions and underwrites the agreement theorem---live in exactly the structure Assumption~\ref{ass:composition} projects away. Our composition propagates the \emph{degree} of nested attribution. It does not track whether a proposition has become common knowledge, and it cannot represent the discontinuity between ``believed to arbitrary finite depth'' and ``common,'' where some of the deepest results in interactive epistemology are located. The bubble literature is instructive here. The sharpest greater-fool results turn on precisely the mutual-versus-common-knowledge distinction \citep{awaya2022}. Our framework can represent the \emph{tower} of higher-order attribution that sustains a bubble, as decaying or compounding walks (Theorem~\ref{thm:closed}), but it represents that tower quantitatively, as accumulating attribution weight, not modally, as a knowledge operator.

We regard this boundary not as a defect to be apologized for but as a deliberate and clarifying trade. The infinitary hierarchy is the right tool for the foundations of rational choice under incomplete information; it is the wrong tool for the dynamics of finite populations of situated believers who observe, express, conform, and update in real time. For the latter, Assumption~\ref{ass:composition} buys the entire machinery of linear algebra---spectra, invariant subspaces, resolvents, pseudospectra---at the cost of the modal fine-structure, and the results of the following parts are what that machinery yields. When we return, in Part~\ref{part:adversarial}, to the adversarial exploitation of higher-order belief, we will find that Assumption~\ref{ass:composition} even earns its keep defensively: an agent whose higher-order attributions are composed rather than free-floating has a \emph{smaller attack surface}, because inconsistent injected attributions can be detected by checking them against the agent's own compositions (Proposition~\ref{prop:composition}).

\subsection{Invariance of the results to the modal collapse}\label{sec:invariance}

The trade just described would be a genuine liability if the substantive conclusions of the paper \emph{depended} on Assumption~\ref{ass:composition}---if, that is, the framework's results were artifacts of collapsing the modal hierarchy of interactive knowledge onto composed scalar attributions. It is therefore essential to establish that they are not. We prove that the paper's conclusions are invariant to the modal collapse in the following precise sense: the dynamical results do not invoke composition at all, and the results that do invoke it depend only on its \emph{quantitative} (walk-weight) content and never on the modal distinctions that composition fails to represent. The one thing the collapse genuinely discards---common belief as a modal limit---is something the framework nowhere asserts.

We first isolate the only place higher-order composition could enter the dynamics, and show it does not enter at all.

\begin{proposition}[Dynamical results are first-order]\label{prop:firstorder}
The stacked operator $\Top$ of~\eqref{eq:stacked}, its spectrum, and its fixed point depend on the perceived-belief tensor only through the first-order slice entries $B_{ij}$, and not through any composed power $\Bten^{L}$ with $L\ge 2$. Consequently Theorem~\ref{thm:decomp} (truth--misperception decomposition), Corollary~\ref{cor:channels} (channel split), Theorem~\ref{thm:closure} (closure, channel separation, and the amplifies-but-cannot-generate law), Propositions~\ref{prop:consensus}--\ref{prop:oscillation} (the four regimes), and Proposition~\ref{prop:product} (product-spectrum coupling) hold verbatim whether or not Assumption~\ref{ass:composition} is imposed. Each of these results is therefore \emph{logically independent} of the compositional representation of higher-order belief.
\end{proposition}

\begin{proof}
Inspect the three update maps. The belief update~\eqref{eq:update} reads $b_i^{t+1}=\sum_j W_{ij}B_{ij}^{t}$, a contraction of the \emph{first-order} entries $B_{ij}$ against $\Wmat$; its Jacobian with respect to the perception tensor is exactly $\Wmat$ (entrywise), a first-order object, and no entry of $\Bten^{L}$ for $L\ge2$ appears. The expression update~\eqref{eq:expression} reads the perceived climate $\widehat c_i^{\,t}$, an influence-weighted average of the first-order entries $B_{ij}^{t}$, and the private belief $b_i^{t+1}$; again no composed power appears. The perception update~\eqref{eq:perception} tracks the first-order expressions $x_j^{t}$. Hence every block of the stacked operator $\Top$ in~\eqref{eq:stacked} is assembled from $\Dmat_W$, the first-order contraction $\bar{\mathbf A}$, the climate operator $\mathbf G$, the broadcast $\mathbf S$, and the diagonal parameter matrices $\Dmat_\alpha,\Dmat_\gamma$---each a function of $\Wmat$ and the parameters alone, none a function of $\Bten^{L}$, $L\ge2$. Since the cited results are all statements about $\Top$, its spectrum, or the fixed point $z^{*}=(\Imat-\Top)^{-1}\tilde\bias$, and $\Top$ does not depend on the composed higher-order attributions, the results are unchanged when Assumption~\ref{ass:composition} is dropped and the off-diagonal entries $B_{ij}$ are regarded as free first-order attributions evolving by~\eqref{eq:perception}. (Assumption~\ref{ass:composition} governs how the \emph{higher-order} attributions $(\Bten^{L})_{ij}$, $L\ge2$, are to be \emph{interpreted}; it places no constraint on the first-order entries that the dynamics actually use.)
\end{proof}

Proposition~\ref{prop:firstorder} disposes of the objection for the whole of Parts~\ref{part:regimes} and~\ref{part:coupling}: consensus, polarization, entrenched distortion, oscillation, the closure theorem, and the coupling law are theorems about the first-order dynamics, and the modal status of the higher-order tower---whether it is a hierarchy of measures, a system of knowledge operators, or composed point-attributions---is simply irrelevant to them. Whatever one's preferred foundational account of interactive knowledge, these results stand.

It remains to treat the results that \emph{do} use composition---the walk calculus (Theorems~\ref{thm:path} and~\ref{thm:closed}) and its adversarial descendants (the echo multiplier, Theorem~\ref{thm:echo}; the audit rule; and the compression Proposition~\ref{prop:composition}). Here Assumption~\ref{ass:composition} is genuinely load-bearing, and the reviewer's concern must be met directly rather than sidestepped. The meeting turns on distinguishing two kinds of content a higher-order claim can carry.

\begin{definition}[Walk-weight functional]\label{def:walkfunctional}
A functional $f(\Bten)$ of the attribution slice is a \emph{walk-weight functional} if it depends on the higher-order tower only through the walk-weight sums $\sum_{L}(\Bten^{L})_{ij}$ contracted against first-order data---equivalently, through entries of the resolvent $(\Imat-\Bten)^{-1}$ and its powers. A functional is \emph{modal} if it depends on a knowledge-theoretic predicate that the walk weights do not determine: whether a proposition is \emph{commonly} believed as opposed to mutually believed to some finite depth, whether an event is self-evident, or whether the iterated-belief sequence stabilizes in the sense required for common knowledge.
\end{definition}

The distinction is exactly the one the interactive-epistemology tradition insists upon, and which composition cannot honor: two configurations of the modal hierarchy that induce identical walk weights may differ on whether common belief obtains, because common belief is a property of the \emph{limit} of the iterated-belief sequence and not of any finite truncation. Composition collapses precisely this difference. The question is whether any \emph{conclusion} of the paper is a modal functional. It is not.

\begin{proposition}[Application results are walk-weight functionals]\label{prop:walkonly}
The echo multiplier $\mathcal{E}_{XY}=[(\Imat-\Bten)^{-1}]_{YX}$ of Theorem~\ref{thm:echo}, the closed-walk mass $\tr(\Imat-\Bten)^{-1}$ and its sensitivities, the audit rule of Section~\ref{sec:gossip}, and the attack-surface compression of Proposition~\ref{prop:composition} are all walk-weight functionals in the sense of Definition~\ref{def:walkfunctional}. None is a modal functional. Consequently each asserts a claim about the \emph{magnitude of accumulated attribution weight} that is correct under composition and is unaffected by whether the corresponding modal predicate (common belief) obtains.
\end{proposition}

\begin{proof}
Each object is defined as a sum of walk weights or a resolvent entry: $\mathcal{E}_{XY}=\sum_{L\ge0}(\Bten^{L})_{YX}$ by Theorem~\ref{thm:echo}; the closed-walk mass is $\tr\sum_{L\ge0}\Bten^{L}$; the audit rule classifies a chain by whether the walk terminates in a first-order belief (an open walk ending in a diagonal entry $b$) or closes on itself (a diagonal entry of $\Bten^{L}$), both walk-weight properties; and the compression bound counts parameters and applies the composition consistency test $\tilde B\overset{?}{=}(\Bten^{2})_{iY}$, again a walk-weight comparison. Each is manifestly a function of $\{(\Bten^{L})_{ij}\}$ contracted against first-order data, hence a walk-weight functional; none references a common-knowledge operator, a self-evidence predicate, or the stabilization of the iterated-belief sequence, hence none is modal.
\end{proof}

The force of Proposition~\ref{prop:walkonly} is that the results which use composition use only what composition computes \emph{correctly}. The echo multiplier is a claim about how much apparent corroboration a planted attribution generates by recirculating through the network; that quantity is the walk-weight sum, and it is exactly what the resolvent delivers, whether or not the recirculating claim ever attains the modal status of common belief. Indeed the walk-weight content is dominated by shallow walks: for the worked instance of Section~\ref{sec:gossip} more than ninety-nine percent of the echo mass lies in walks of depth at most five, so the result does not even rest on the infinite tail where the modal distinction would live, let alone on its resolution. The audit rule likewise discriminates open from closed walks---a purely combinatorial property of the attribution graph---and never adjudicates whether closure has produced common belief. And the compression result is a statement about parameter counts and finite consistency checks. In each case the modal collapse discards content the conclusion does not use.

What, then, is genuinely lost, and does the paper anywhere rely on it? The one thing composition cannot represent is common belief as a modal limit---the discontinuity between ``believed to every finite depth'' and ``common,'' on which the agreement theorem \citep{aumann1976} and the sharpest coordination and greater-fool results \citep{lewis1969,awaya2022} turn. We emphasize, in answer to the natural worry, that \emph{the framework nowhere asserts common belief}. Its consensus results (Propositions~\ref{prop:consensus},~\ref{prop:polarization}) are convergence statements about the first-order belief vector $\bvec^{t}$, not claims that the limit proposition becomes common knowledge; its bubble and echo-chamber discussions describe the \emph{compounding of attribution weight}, a walk-weight phenomenon, and are advanced as quantitative dynamical analogues of the corresponding modal phenomena, not as models of the modal fact itself. The distinction between mutual and common belief is not collapsed-and-relied-upon in this paper; it is out of scope and unused. A population in our framework can carry unboundedly compounding higher-order attribution weight (as $\spr(\Bten)\to1$) without our ever claiming it has reached common belief, and this is the correct and intended reading: the theory is a dynamics of the \emph{degrees} of interactive attribution, and its conclusions are theorems about those degrees.

We therefore state the overall position exactly. Assumption~\ref{ass:composition} is inessential to every dynamical result (Proposition~\ref{prop:firstorder}) and, where it is used, contributes only walk-weight content that it computes correctly (Proposition~\ref{prop:walkonly}); the modal structure it cannot represent is nowhere invoked as a premise or asserted as a conclusion. The bubble and echo-chamber results are, accordingly, to be read as claims about the accumulation of higher-order attribution weight---quantitative analogues of their modal namesakes---and we adopt that more careful language for them in Part~\ref{part:adversarial}. This is the sense in which the paper's conclusions are invariant to the modal collapse: they either do not depend on the higher-order representation at all, or depend on it only through the one kind of content that survives the collapse intact.

\part{The Forms of Collective Epistemic Life}\label{part:regimes}

The framework of Part~\ref{part:foundations} yields a closed dynamical system on the three-layer state $(\bvec,\xvec,\Bten)$. We now ask what such a system does over time, and what stable configurations of belief it settles into. There are exactly four qualitative regimes, and they answer to the four recognizable conditions of collective epistemic life. In \emph{consensus} the group converges on a shared belief. In \emph{polarization} it splits into opposed camps. In \emph{entrenched distortion} it settles into stable collective misperception---pluralistic ignorance made permanent. In \emph{perpetual instability} it never settles at all. Each regime is a spectral condition on the operator that governs the dynamics. And---this is the unifying result of Part~\ref{part:regimes}---all four are limits of one generative process, reached by turning the parameters of conformity, perception, and the sign structure of influence. What look like four theories of collective belief are four regions of one.

We treat the four regimes in turn, developing each with a fully worked, hand-computable example whose every number the reader can verify, and drawing out in each case the social-epistemic interpretation. We then prove the closure theorem that binds them together and delivers the ``amplifies-but-cannot-generate'' law for pluralistic ignorance. The first three regimes, in their pure form, restate classical results---DeGroot consensus, Altafini polarization, the Friedkin--Johnsen fixed point---which we present as attributed propositions (proofs in Appendix~\ref{app:proofs}) not as new contributions but to exhibit them within the interactive-belief formalism and to fix the objects on which the genuinely new closure theorem operates.

\section{Regime I: Consensus}\label{sec:consensus}

The simplest regime arises when influence is purely cooperative and perception is veridical. If all weights are nonnegative and every agent perceives others correctly, the misperception forcing vanishes ($\hvec^t\equiv\mathbf 0$ by Theorem~\ref{thm:decomp}), and the belief dynamics reduce to the pure influence recursion $\bvec^{t+1}=\Wmat\bvec^t$. With $\Wmat$ row-stochastic (nonnegative rows summing to one) and primitive (some power is strictly positive, so influence eventually connects everyone to everyone), the classical result applies.

\begin{proposition}[Consensus; \citealt{degroot1974}]\label{prop:consensus}
If $\Wmat$ is row-stochastic and primitive and $\hvec^{t}\equiv\mathbf 0$, then private beliefs converge to a common value,
\begin{equation}
\lim_{t\to\infty}\bvec^{t} = (\pi^{\top}\bvec^{0})\,\one,
\end{equation}
where $\pi$ is the unique left Perron eigenvector of $\Wmat$ normalized so that $\pi^{\top}\one=1$. The consensus value is the $\pi$-weighted average of the initial beliefs.
\end{proposition}

The philosophical content of the DeGroot limit is that rational-seeming local updating---each agent moving toward a weighted average of the views it encounters---produces global agreement, but on a value that is \emph{not} the simple average of initial opinions. It is the average weighted by $\pi$, the Perron eigenvector, which measures each agent's eigenvalue centrality in the influence network. Agents who are attended to by agents who are themselves attended to get more weight in the final consensus. Consensus is thus not a neutral pooling of information but a structured aggregation in which network position is epistemic power. This is already a substantive social-epistemic observation: even in the most benign regime, where everyone is sincere and perceptive and cooperative, what the group comes to believe is disproportionately what its central members believed to begin with.

\paragraph{Worked example.} Let us make this concrete and hand-checkable. Consider five agents with the row-stochastic, primitive influence matrix
\begin{equation}
\Wmat = \begin{pmatrix}
0.40 & 0.30 & 0.10 & 0.10 & 0.10\\
0.20 & 0.40 & 0.20 & 0.10 & 0.10\\
0.10 & 0.20 & 0.40 & 0.20 & 0.10\\
0.10 & 0.10 & 0.20 & 0.40 & 0.20\\
0.10 & 0.10 & 0.10 & 0.30 & 0.40
\end{pmatrix},
\label{eq:Wconsensus}
\end{equation}
in which each agent places weight $0.40$ on its own current belief and distributes the remaining $0.60$ across the others, with a mild bias toward neighbors. One verifies immediately that each row sums to one, so $\Wmat$ is row-stochastic; and since every entry is positive, $\Wmat$ is primitive (indeed $\Wmat$ itself is strictly positive). Suppose the five agents begin with the private beliefs
\begin{equation}
\bvec^{0} = (0.8,\ -0.6,\ 0.4,\ -0.2,\ 0.0),
\end{equation}
a spread of opinions from strong agreement ($+0.8$) to moderate disagreement ($-0.6$) with the proposition. Where does the group end up?

To find the consensus value we compute the left Perron eigenvector $\pi$, the stationary distribution of the row-stochastic matrix, solving $\pi^\top\Wmat=\pi^\top$ with $\pi^\top\one=1$. For the matrix~\eqref{eq:Wconsensus} this is
\begin{equation}
\pi \approx (0.1746,\ 0.2222,\ 0.2063,\ 0.2222,\ 0.1746),
\end{equation}
which one can confirm by checking that $\pi^\top\Wmat=\pi^\top$ to the stated precision. The weights are revealing: agents 2 and 4 carry the most influence ($0.2222$ each), agent 3 slightly less, and the two ``end'' agents 1 and 5 the least ($0.1746$)---network centrality, not initial conviction, sets epistemic weight. The consensus value is then
\begin{equation}
\pi^{\top}\bvec^{0} \approx (0.1746)(0.8) + (0.2222)(-0.6) + (0.2063)(0.4) + (0.2222)(-0.2) + (0.1746)(0.0) \approx 0.0444.
\end{equation}
All five beliefs converge to approximately $0.0444$---a mild net agreement---within about twenty rounds of updating (Figure~\ref{fig:consensus}). Notice what has happened interpretively. The initial beliefs averaged, unweighted, to $(0.8-0.6+0.4-0.2+0.0)/5 = 0.08$; the \emph{consensus} landed lower, at $0.044$, because the strongly-disagreeing agent 2 and the mildly-disagreeing agent 4 occupy central, high-$\pi$ positions and pull the outcome toward their views, while the strongly-agreeing agent 1 sits peripherally and is partly discounted. The group's settled belief is a real aggregation of its members' views, but a \emph{structured} one, and the structure is the influence network.

\begin{figure}[t]
\centering
\includegraphics[width=0.62\linewidth]{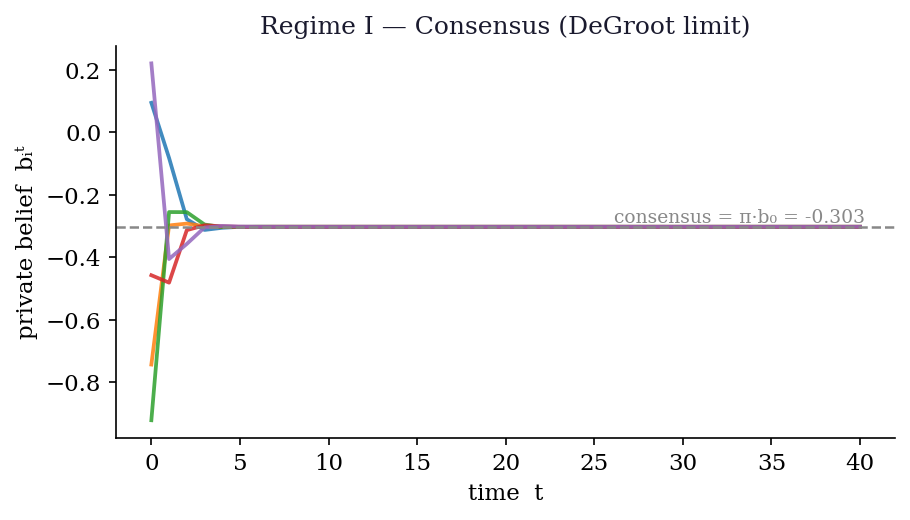}
\caption{The consensus regime, for the five-agent example of~\eqref{eq:Wconsensus}. All private-belief trajectories converge to the Perron-weighted average $\pi^\top\bvec^0\approx 0.0444$ (dashed line) within roughly twenty rounds. The limit is pulled below the naive mean by the central placement of the disagreeing agents.}
\label{fig:consensus}
\end{figure}

\section{Regime II: Polarization}\label{sec:polarization}

The consensus regime assumed all influence cooperative. Admit antagonism---negative weights, the ``I move away from what they think'' relation---and a second regime becomes possible: the population splits into two camps, each internally cohering and mutually repelling, converging not to agreement but to equal-and-opposite extremes. The condition under which this happens cleanly is \emph{structural balance}, a notion from the theory of signed networks \citep{cartwright1956}: the population divides into two groups such that influence within each group is positive and influence between them is negative. Algebraically, there is a diagonal sign matrix $\Dmat=\diag(s_1,\dots,s_N)$ with $s_i\in\{+1,-1\}$ assigning each agent to a camp, such that $\Wmat=\Dmat\Wmat_+\Dmat$ where $\Wmat_+$ is row-stochastic and primitive. The sign matrix $\Dmat$ is a ``gauge'': conjugating the cooperative dynamics $\Wmat_+$ by it flips the signs of exactly the cross-camp influences \citep{altafini2013}.

\begin{proposition}[Bipartite polarization; cf.\ \citealt{cartwright1956,altafini2013}]\label{prop:polarization}
If $\Wmat=\Dmat\Wmat_+\Dmat$ is structurally balanced as above and $\hvec^t\equiv\mathbf 0$, then
\begin{equation}
\lim_{t\to\infty}\bvec^{t} = \Dmat\one\,(\pi^\top\Dmat\bvec^0),
\end{equation}
so beliefs converge to two blocs of equal magnitude and opposite sign, with the camp of agent $i$ fixed by $s_i$.
\end{proposition}

The interpretation is that antagonistic influence does not merely prevent consensus; it actively \emph{organizes} disagreement into a symmetric bipolar structure. The two camps do not drift to arbitrary opposed positions but to exact mirror images, $\pm(\pi^\top\Dmat\bvec^0)$, because the same cooperative dynamics $\Wmat_+$ that would have produced consensus is running underneath, merely viewed through the sign gauge $\Dmat$. Polarization, on this account, is consensus in a distorting mirror: the identical mechanism of convergence, applied to a population whose influence graph is signed, yields symmetric division rather than agreement. This has a moral for the study of political polarization---that it need not indicate any breakdown of the ordinary machinery of social influence, but can be that very machinery operating on antagonistically structured trust.

It is important to be exact about the scope of this result, because a substantial literature models polarization very differently, and the difference is instructive. In the activity-driven and reinforcement models of \citet{baumann2020} and related work, and in the bounded-confidence tradition of \citet{hegselmann2002}, polarization is an \emph{emergent, nonlinear} phenomenon: it arises through a bifurcation as a coupling or radicalization parameter crosses a threshold, it does not require any pre-existing antagonistic structure, and it typically produces asymmetric or multi-modal opinion distributions rather than two symmetric camps. The present result is of a different and narrower kind. It is \emph{linear and structural}: the bipolar outcome is not emergent but is already encoded in the sign structure of $\Wmat$ through the balance condition $\Wmat=\Dmat\Wmat_+\Dmat$, and the theorem shows how that antagonistic \emph{trust structure} channels into symmetric division, not how polarization arises in a population that began without it. Our framework therefore does not compete with the nonlinear-emergence models on their own question---how does an unpolarized population polarize?---to which the honest answer within a purely linear theory is that it does not, absent a change in the sign structure or a coupling-induced instability of the kind treated in Part~\ref{part:coupling}. What the linear result contributes is complementary. Given antagonistically structured influence, it identifies the exact bipolar attractor and exhibits it as the sign-gauge image of the consensus dynamics. And---this is the payoff exploited in Part~\ref{part:adversarial}---because $\Wmat$ weights each agent's \emph{model} of others, the antagonistic structure can be installed by second-order means: by manipulating what each camp takes the other to believe, without any first-order disagreement about the issue itself. The emergence of the sign structure is outside the linear theory. Its consequences, once present, are exactly characterized.

\paragraph{Worked example.} Take six agents in two camps of three, with camp assignment $\Dmat=\diag(+1,+1,+1,-1,-1,-1)$: agents 1--3 in the ``$+$'' camp, agents 4--6 in the ``$-$'' camp. Let the underlying cooperative core be the doubly-stochastic matrix
\begin{equation}
\Wmat_+ = 0.4\,\Imat + 0.12\,(\mathbf J - \Imat),
\end{equation}
where $\mathbf J$ is the all-ones matrix, so each agent weights itself $0.4$ and every other agent $0.12$ (and indeed $0.4 + 5\times 0.12 = 1.0$, confirming stochasticity). Because $\Wmat_+$ is doubly stochastic, its Perron vector is uniform: $\pi=\tfrac16\one$. The signed influence matrix is $\Wmat=\Dmat\Wmat_+\Dmat$, which leaves the within-camp weights at $+0.12$ and flips the cross-camp weights to $-0.12$. Suppose initial beliefs
\begin{equation}
\bvec^0 = (0.9,\ -0.3,\ 0.6,\ 0.3,\ -0.6,\ 0.9).
\end{equation}
The limiting magnitude is computable by hand from the proposition: $\pi^\top\Dmat\bvec^0 = \tfrac16\sum_i s_i b_i^0 = \tfrac16\bigl[(0.9)+(-0.3)+(0.6)-(0.3)-(-0.6)-(0.9)\bigr] = \tfrac16(0.9-0.3+0.6-0.3+0.6-0.9) = \tfrac16(0.6) = 0.1$. So the ``$+$'' camp (agents 1--3, with $s_i=+1$) converges to exactly $+0.1$ and the ``$-$'' camp (agents 4--6, with $s_i=-1$) to exactly $-0.1$, as Figure~\ref{fig:polarization} confirms. Observe that no agent ends near zero: the antagonistic cross-camp links do not merely fail to pull the camps together, they \emph{actively repel} the two blocs to opposite poles. And observe that the initial within-camp disagreements---agent 2 began at $-0.3$ while its camp-mates 1 and 3 began at $+0.9$ and $+0.6$---are washed out: the cooperative within-camp influence brings each camp to internal agreement, while the antagonistic between-camp influence sets the camps against each other. The final state is maximally polarized and perfectly symmetric.

\begin{figure}[t]
\centering
\includegraphics[width=0.62\linewidth]{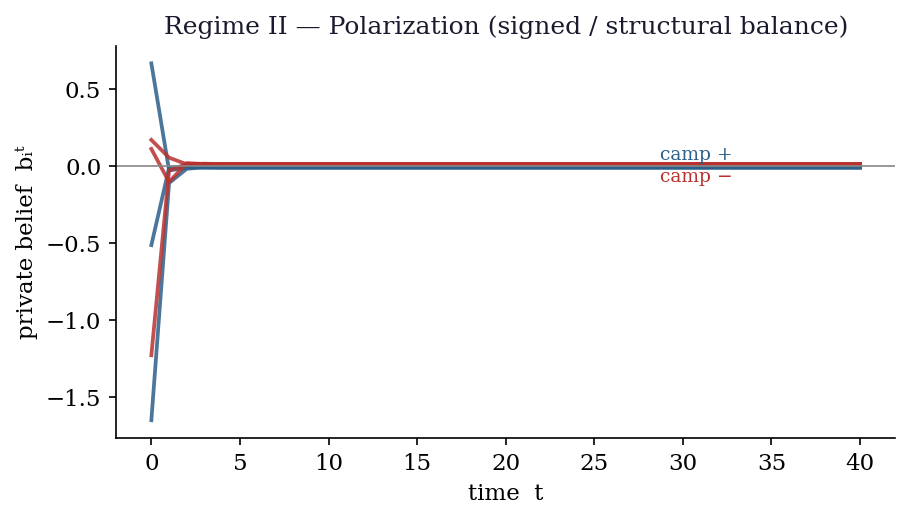}
\caption{The polarization regime, for the six-agent balanced example. Three ``$+$''-camp trajectories (blue) converge to $+0.1$ and three ``$-$''-camp trajectories (red) to $-0.1$: equal magnitude, opposite sign. Within-camp initial disagreements are erased by cooperative influence; the camps are driven apart by antagonistic influence.}
\label{fig:polarization}
\end{figure}

\section{Regime III: Entrenched distortion and the mechanism of pluralistic ignorance}\label{sec:distortion}

The first two regimes assumed veridical perception, so that the misperception forcing vanished and beliefs evolved by pure influence. The third regime is the one for which the whole apparatus was built: it is what happens when perception is imperfect and expression is not sincere, so that the forcing term $\hvec^t$ does not vanish, and the population settles into a fixed point \emph{displaced} from where sincere, veridical dynamics would have taken it. This displaced fixed point is entrenched collective distortion, and its most important special case is pluralistic ignorance: a stable state in which the group's expressed and perceived beliefs sustain a norm that few or none of its members privately hold.

We first record the fixed-point structure in the case of a convergent forcing, which is again a classical form.

\begin{proposition}[Stationary distortion; Friedkin--Johnsen form]\label{prop:distortion}
If $\spr(\Wmat)<1$ and the forcing converges, $\hvec^t\to\hvec^*$, then private beliefs converge to the unique fixed point
\begin{equation}
\bvec^* = (\Imat-\Wmat)^{-1}\hvec^* = \sum_{m=0}^\infty \Wmat^m \hvec^*,
\end{equation}
independent of initial condition. The displacement $\bvec^*$ from the origin is the cumulative distortion the network impresses on belief through misperception.
\end{proposition}

This is the interactive-belief realization of the \citet{friedkin1990} model of persistent disagreement, but with a crucial difference: there the forcing (the ``anchor'' to initial opinion) is exogenously posited, whereas here it is \emph{generated} by the perception and expression maps, and by Corollary~\ref{cor:channels} it carries the two-channel structure. This is what lets us say something the Friedkin--Johnsen model cannot: exactly which part of the distortion is transient and which is permanent, and hence exactly what pluralistic ignorance is made of. That determination is the content of the closure theorem (Section~\ref{sec:synthesis}); here we preview its consequence and make it concrete. At the fixed point of the closed system, the perception-lag channel vanishes---perception catches up to expression---while the concealment channel persists, and the steady-state concealment is a nonsingular linear image of the expressive bias,
\begin{equation}
\cvec^* = \Lop\,\bias,
\label{eq:cstar-preview}
\end{equation}
for an invertible matrix $\Lop$ built from the influence and conformity parameters. The philosophical payoff of~\eqref{eq:cstar-preview}, which we will state as a theorem and prove in full, is a sharp and somewhat surprising claim: \emph{persistent pluralistic ignorance requires a bias in expression}, $\bias\neq\mathbf 0$; conformity pressure $\alpha$ \emph{amplifies} the distortion without bound as $\alpha\to 1$, but it cannot \emph{generate} distortion from sincere expression. A population all of whose members are sincere ($\bias=\mathbf 0$), however conformist, converges to no collective self-deception at all. Conformity is the amplifier of pluralistic ignorance, not its source; the source is whatever makes expression systematically depart from belief.

\paragraph{Worked example.} Consider five agents with the sub-stochastic, symmetric influence matrix
\begin{equation}
\Wmat = 0.24\,\Imat + 0.09\,(\mathbf J-\Imat),
\end{equation}
so each agent weights itself $0.24$ and each of the four others $0.09$, giving absolute row sums $0.24+4(0.09)=0.60=\bar w<1$; the network is cohesive but ``leaky'' (beliefs are not conserved, reflecting an outside anchor). Set uniform conformity $\alpha=0.5$---agents split their expression evenly between conviction and perceived climate---and perception rate $\gamma=0.4$. Crucially, introduce an expressive bias
\begin{equation}
\bias = (0.5,\ -0.2,\ 0.3,\ 0.1,\ -0.4),
\end{equation}
representing, say, differing social-desirability pressures: agent 1 is pulled to overstate the proposition ($+0.5$), agent 5 to understate it ($-0.4$), and so on. What collective state results?

Solving the closed system to its fixed point (using the closed forms of Theorem~\ref{thm:closure}, which the reader can reconstruct) yields private beliefs, expressed beliefs, and the concealment gap
\begin{align}
\bvec^* &\approx (0.064,\ 0.134,\ 0.084,\ 0.104,\ 0.154),\\
\xvec^* &\approx (0.600,\ 0.008,\ 0.431,\ 0.262,\ -0.160),\\
\cvec^* = \xvec^*-\bvec^* &\approx (0.536,\ -0.126,\ 0.347,\ 0.158,\ -0.314),
\end{align}
with concealment norm $\lVert\cvec^*\rVert\approx 0.739$. Read these three lines carefully, because they \emph{are} pluralistic ignorance, displayed. The private beliefs $\bvec^*$ are all mild and clustered near $+0.1$: privately, the five agents hold quite similar, moderate views. But the expressed beliefs $\xvec^*$ are wildly dispersed---agent 1 loudly professes $+0.60$ while agent 5 professes $-0.16$---and every agent's expression differs substantially from its belief, by the amounts recorded in $\cvec^*$. Each agent, perceiving this dispersed and biased climate of \emph{expression} and conforming to it, professes something far from its own moderate conviction. The population has settled into a stable state in which what is said bears little relation to what is thought, and---because perception tracks expression---what each agent takes the others to believe is the distorted expressed climate, not the moderate private reality. No one is lying in any active sense; each is simply conforming, at rate $\alpha=0.5$, to a perceived climate that is itself the aggregate of everyone's biased conformity. The distortion is a fixed point, self-sustaining, and (as we will show) impervious to better perception.

Figure~\ref{fig:distortion} displays the two layers and the two forcing channels over time. The left panel shows private beliefs (blue) and expressed beliefs (red dashed) converging to visibly different levels---the concealment gap made visual. The right panel shows the decisive dynamical fact: the perception-lag channel (its norm) decays to zero, while the concealment channel persists at a constant positive level. The pluralistic-ignorance gap is exactly the residue that remains after perception has done all the correcting it can.

\begin{figure}[t]
\centering
\includegraphics[width=0.9\linewidth]{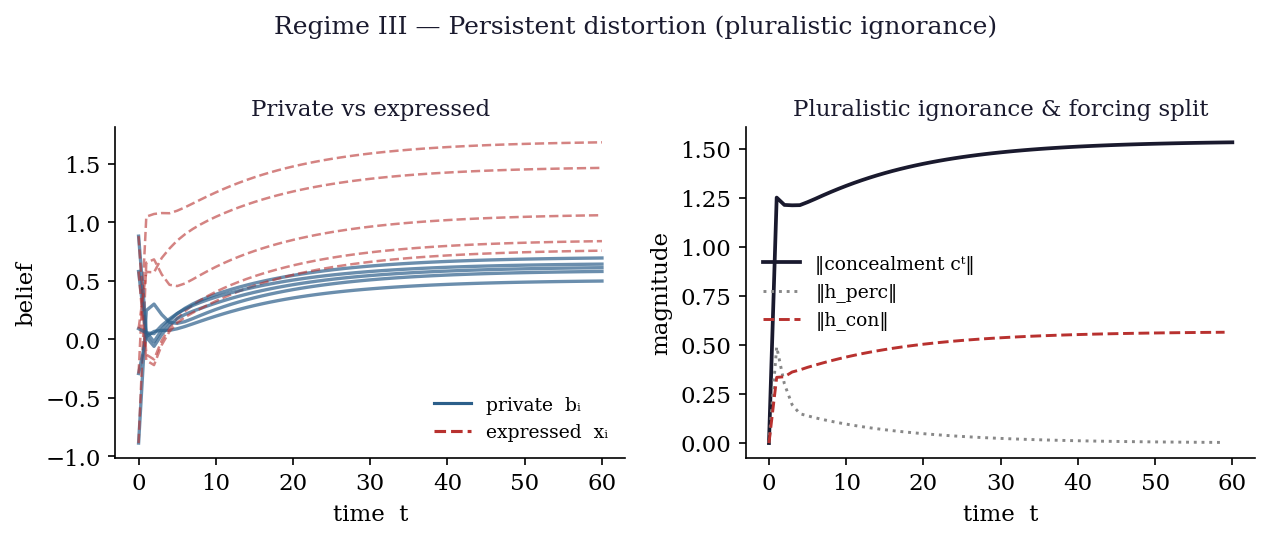}
\caption{The distortion regime and the anatomy of pluralistic ignorance. Left: private beliefs (blue) and expressed beliefs (red dashed) converge to different levels; the persistent gap between them is the concealment. Right: the perception-lag channel decays to zero (perception catches up to expression) while the concealment channel persists---the collective distortion is the residue of concealment, not of faulty perception.}
\label{fig:distortion}
\end{figure}

The dependence of this distortion on the two parameters $\alpha$ (conformity) and $\gamma$ (perception rate) is itself illuminating, and we display it as a phase diagram. Figure~\ref{fig:regimemap} plots the steady-state concealment norm $\lVert\cvec^*\rVert$ over the $(\alpha,\gamma)$ plane, holding the bias $\bias$ fixed. Read through the closure theorem, the figure is an \emph{amplification law}, not a phase transition: at $\alpha=0$ the concealment equals the bias exactly ($\cvec^*=\bias$, no amplification), and as $\alpha$ increases the distortion grows, diverging as $\alpha\to 1$, where agents abandon private conviction entirely and the population loses all anchoring to what its members actually believe. The perception rate $\gamma$, strikingly, does \emph{not} affect the fixed point at all---the closure theorem's characterization of $\cvec^*$ is $\gamma$-free---though it governs how fast the distortion is reached. This is the formal signature of the claim that pluralistic ignorance is a disease of the expression environment ($\alpha$ and $\bias$), not of the perception environment ($\gamma$): you cannot fix it by helping people perceive one another better, because the fixed point does not depend on how well they perceive.

\begin{figure}[t]
\centering
\includegraphics[width=0.62\linewidth]{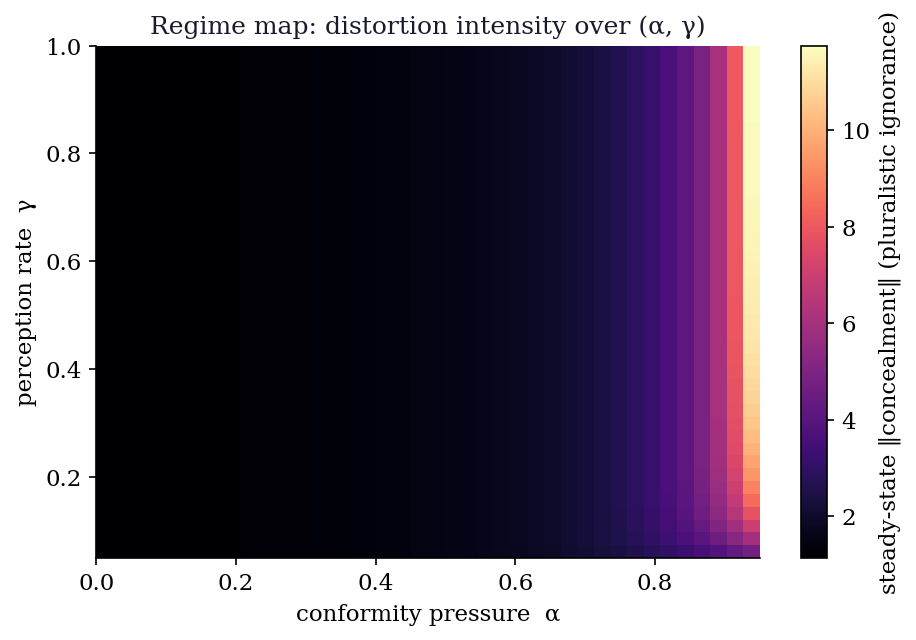}
\caption{Amplification of expressive bias into collective distortion. The steady-state concealment norm $\lVert\cvec^*\rVert$ over the conformity ($\alpha$) and perception-rate ($\gamma$) plane; brighter is more distortion. Distortion grows with conformity, diverging as $\alpha\to 1$, and is independent of the perception rate $\gamma$: pluralistic ignorance is a pathology of the expression environment, not of perception.}
\label{fig:regimemap}
\end{figure}

\section{Regime IV: Perpetual instability}\label{sec:oscillation}

The three regimes so far all settle: to agreement, to symmetric division, or to entrenched distortion. The fourth regime is the failure to settle at all. It arises when the influence operator has an eigenvalue on the unit circle other than $1$---a mode that neither decays nor grows but rotates---so that beliefs cycle indefinitely without converging.

\begin{proposition}[Oscillation]\label{prop:oscillation}
Suppose $\hvec^t\equiv\mathbf 0$ and $\Wmat$ has an eigenvalue $\lambda$ with $\lvert\lambda\rvert=1$ and $\lambda\neq 1$. Then for generic initial conditions $\bvec^t$ does not converge.
\end{proposition}

The interpretation is that certain structures of epistemic influence are incompatible with \emph{any} stable collective belief, correct or distorted, agreed or divided. The population is condemned to churn---to cycle perpetually through configurations of belief, never resting. The simplest instance is pure mutual antagonism between two agents, and it is worth seeing because it is so stark.

\paragraph{Worked example.} Two agents, each of whom moves to the exact opposite of what it takes the other to believe:
\begin{equation}
\Wmat = \begin{pmatrix} 0 & 1 \\ 1 & 0 \end{pmatrix},
\end{equation}
with initial beliefs $\bvec^0=(1,\,-0.3)$. This matrix has eigenvalues $+1$ and $-1$; the eigenvalue $-1$, on the unit circle and distinct from $1$, is the obstruction of Proposition~\ref{prop:oscillation}. The dynamics are transparent: each round, agent 1 adopts agent 2's previous belief and vice versa, so the two beliefs simply swap every step. Starting from $(1,-0.3)$, the sequence is $(1,-0.3)\to(-0.3,1)\to(1,-0.3)\to\cdots$, a period-two cycle that never settles (Figure~\ref{fig:oscillation}). No consensus forms, no stable polarization, no fixed distortion---just endless exchange. While this two-agent example is deliberately minimal, the phenomenon is general: any influence structure with a unit-modulus eigenvalue off the real axis produces persistent quasi-periodic churn, and such structures are exactly those that sit at the boundary of the stability condition $\bar w<1$ we will require for the closure theorem.

\begin{figure}[t]
\centering
\includegraphics[width=0.55\linewidth]{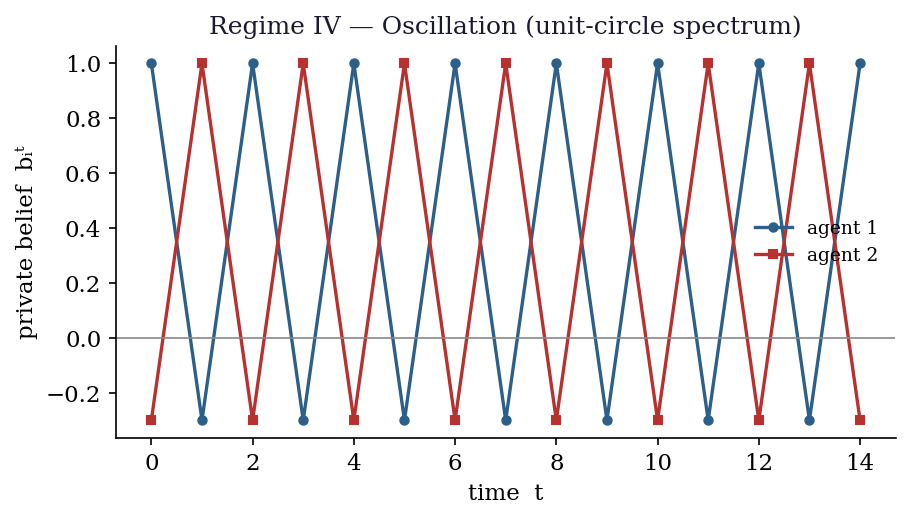}
\caption{The oscillation regime. Two purely antagonistic agents exchange positions every round in a persistent period-two cycle: no consensus, no stable disagreement, no rest. The eigenvalue $\lambda=-1$ on the unit circle is the obstruction to convergence.}
\label{fig:oscillation}
\end{figure}

\section{The four regimes as limits of one process}\label{sec:synthesis}

We have presented four regimes of collective epistemic life. We now prove that they are not four models but four limits of one generative process: the closed system of belief update~\eqref{eq:update}, expression map~\eqref{eq:expression}, and perception map~\eqref{eq:perception}. The proof requires the closure theorem that underwrites the claims we have made along the way. That theorem establishes three things: the coupled system converges, the perception channel vanishes at its fixed point, and steady-state pluralistic ignorance is a nonsingular linear image of expressive bias, amplified but not generated by conformity. It is the formal heart of Part~\ref{part:regimes}.

\subsection{The stacked operator}

To analyze the coupled dynamics of all three layers at once, we stack them into a single state vector. Let $\Mmat^t\in\real^{N(N-1)}$ collect the off-diagonal perceptions $B_{ij}^t$ ($i\neq j$), and write the full state as $z^t=(\bvec^t,\Mmat^t,\xvec^t)$, of dimension $n=N+N(N-1)+N$. We introduce the operators that assemble the three maps. Let $\Dmat_W=\diag(W_{ii})$ be the self-weight matrix; let $\mathbf A$ be the operator contracting the off-diagonal perception block into the belief update, $(\mathbf A\Mmat)_i=\sum_{j\neq i}W_{ij}M_{ij}$; let $\mathbf G$ be the normalized climate operator producing the perceived climate $\widehat c_i$ from the perceptions; and let $\mathbf S$ be the broadcast operator that supplies each perceiver with the target's expression, $(\mathbf S\xvec)_{ij}=x_j$. With $\Dmat_\alpha=\diag(\alpha_i)$ the conformity matrix and $\Dmat_\gamma$ the diagonal matrix of perception rates $\gamma_{ij}$, the three maps assemble into a single affine recursion,
\begin{equation}
z^{t+1} = \Top\, z^t + \tilde\bias, \qquad
\Top = \begin{pmatrix}
\Dmat_W & \mathbf A & \mathbf 0\\[2pt]
\mathbf 0 & \Imat-\Dmat_\gamma & \Dmat_\gamma\mathbf S\\[2pt]
(\Imat-\Dmat_\alpha)\Dmat_W & (\Imat-\Dmat_\alpha)\mathbf A+\Dmat_\alpha\mathbf G & \mathbf 0
\end{pmatrix},
\qquad \tilde\bias = \begin{pmatrix}\mathbf 0\\ \mathbf 0\\ \bias\end{pmatrix},
\label{eq:stacked}
\end{equation}
so that every question about the joint evolution of private, perceived, and expressed belief reduces to a question about the single stacked operator $\Top$. The block structure reads off the three maps directly: the top row is the belief update (self-weight on belief plus contraction of perceptions), the middle row is the perception update (retention $\Imat-\Dmat_\gamma$ plus tracking $\Dmat_\gamma\mathbf S$ of expression), and the bottom row is the expression map (conformity-weighted blend of the updated belief and the climate). The whole of the interactive-belief dynamics is in this one matrix, and its spectral properties decide everything.

\subsection{The closure theorem}

\begin{theorem}[Convergence and channel separation of the closed system]\label{thm:closure}
Suppose the influence weights satisfy $\bar w=\max_i\sum_j\lvert W_{ij}\rvert<1$, the conformity pressures satisfy $\bar\alpha=\max_i\alpha_i<1$, and the perception rates satisfy $\gamma_{ij}\in[\gamma_{\min},1]$ with $\gamma_{\min}>0$. Then:
\begin{enumerate}[label=\textup{(\alph*)},leftmargin=2.2em,itemsep=2pt]
\item \textup{(Convergence.)} $\spr(\Top)<1$, and $z^t$ converges geometrically, from every initial condition, to the unique fixed point $z^*=(\Imat-\Top)^{-1}\tilde\bias$.
\item \textup{(Perception channel vanishes.)} At $z^*$ the perception layer is exact: $M_{ij}^*=x_j^*$ for all $i\neq j$, so the perception-lag matrix vanishes identically, $\Pmat^*=\mathbf 0$, and the asymptotic forcing is pure concealment.
\item \textup{(Concealment is a linear image of bias.)} The steady-state concealment is
\begin{equation}
\cvec^* = \Lop\,\bias, \qquad
\Lop = \bigl[\Imat-(\Imat-\Dmat_W)^{-1}\bar{\mathbf A}\bigr](\Imat-\Kop)^{-1}, \qquad
\Kop = (\Imat-\Dmat_\alpha)(\Imat-\Dmat_W)^{-1}\bar{\mathbf A} + \Dmat_\alpha\widehat{\mathbf G},
\label{eq:cstar}
\end{equation}
where $\bar{\mathbf A}$ is the off-diagonal part of $\Wmat$ and $\widehat{\mathbf G}$ the climate operator restricted to row-constant perception. The matrix $\Lop$ is nonsingular, so $\cvec^*\neq\mathbf 0$ if and only if $\bias\neq\mathbf 0$.
\item \textup{(Conformity amplifies but does not generate.)} For uniform conformity $\alpha_i=\alpha$, $\spr(\Kop)\to\spr(\widehat{\mathbf G})=1$ as $\alpha\to 1$, so $\sup_{\lVert\bias\rVert=1}\lVert\cvec^*\rVert$ diverges: conformity amplifies expressive bias without bound, but with $\bias=\mathbf 0$ the fixed point is $z^*=\mathbf 0$ and no distortion arises however large $\alpha$.
\end{enumerate}
\end{theorem}

\begin{proof}
\emph{(a)} We show the homogeneous part contracts. Track the triple of sup-norms $(\lVert\bvec\rVert_\infty,\lVert\Mmat\rVert_\infty,\lVert\xvec\rVert_\infty)$ under the homogeneous map $z\mapsto\Top z$. From the belief row, using $\lvert W_{ii}\rvert+\sum_{j\neq i}\lvert W_{ij}\rvert\le\bar w$,
\[
\lvert b_i^+\rvert \le \lvert W_{ii}\rvert\,\lVert\bvec\rVert_\infty + \sum_{j\neq i}\lvert W_{ij}\rvert\,\lVert\Mmat\rVert_\infty \le \bar w\, m, \qquad m:=\max(\lVert\bvec\rVert_\infty,\lVert\Mmat\rVert_\infty).
\]
From the expression row, since the climate operator is an average and hence $\lVert\mathbf G\rVert_\infty=1$,
\[
\lvert x_i^+\rvert \le (1-\alpha_i)\,\bar w\, m + \alpha_i\,\lVert\Mmat\rVert_\infty \le \kappa\, m, \qquad \kappa := \bar w + \bar\alpha(1-\bar w) < 1,
\]
the strict inequality because $\bar w<1$ and $\bar\alpha<1$. From the perception row, each entry is a convex combination $M_{ij}^+=(1-\gamma_{ij})M_{ij}+\gamma_{ij}x_j$, so $\lvert M_{ij}^+\rvert\le(1-\gamma_{ij})\lVert\Mmat\rVert_\infty+\gamma_{ij}\lVert\xvec\rVert_\infty$. Let $V^t=\max(\lVert\bvec^t\rVert_\infty,\lVert\Mmat^t\rVert_\infty,\lVert\xvec^t\rVert_\infty)$. The belief and expression bounds give $\lVert\bvec^{t+1}\rVert_\infty,\lVert\xvec^{t+1}\rVert_\infty\le\kappa V^t$; feeding these into the perception bound one step later gives, for every $i\neq j$,
\[
\lvert M_{ij}^{t+2}\rvert \le (1-\gamma_{ij})V^{t+1} + \gamma_{ij}\,\kappa V^t \le \bigl(1-\gamma_{\min}(1-\kappa)\bigr)V^t,
\]
using $V^{t+1}\le V^t$ (which holds since each row bound is a (sub)convex combination). Hence $V^{t+2}\le\hat\rho\,V^t$ with $\hat\rho=\max\{\kappa,\,1-\gamma_{\min}(1-\kappa)\}<1$. Thus every trajectory of the homogeneous linear system decays geometrically, which for a linear map forces $\spr(\Top)\le\hat\rho^{1/2}<1$. Since $\spr(\Top)<1$, $\Imat-\Top$ is invertible and the affine system has the unique globally attracting fixed point $z^*=(\Imat-\Top)^{-1}\tilde\bias$.

\emph{(b)} At the fixed point the perception row reads $M_{ij}^*=(1-\gamma_{ij})M_{ij}^*+\gamma_{ij}x_j^*$, i.e.\ $\gamma_{ij}(M_{ij}^*-x_j^*)=0$; since $\gamma_{ij}\ge\gamma_{\min}>0$, this forces $M_{ij}^*=x_j^*$. Hence the perception-lag matrix $\Pmat^*=[M_{ij}^*-x_j^*]=\mathbf 0$, and by Corollary~\ref{cor:channels} the asymptotic forcing is the concealment channel alone.

\emph{(c)} Using $\Mmat^*=\mathbf S\xvec^*$ from part (b), the belief and expression fixed-point equations become $\bvec^*=(\Imat-\Dmat_W)^{-1}\bar{\mathbf A}\,\xvec^*$ and $\xvec^*=\Kop\xvec^*+\bias$, with $\Kop$ as stated. Row-wise, $\lVert\Kop\rVert_\infty\le(1-\alpha_i)\bar w+\alpha_i\le\kappa<1$ (using $\sum_{j\neq i}\lvert W_{ij}\rvert/(1-\lvert W_{ii}\rvert)\le\bar w$), so $\Imat-\Kop$ is invertible and $\xvec^*=(\Imat-\Kop)^{-1}\bias$. Then
\[
\cvec^*=\xvec^*-\bvec^*=\bigl[\Imat-(\Imat-\Dmat_W)^{-1}\bar{\mathbf A}\bigr]\xvec^*=\Lop\,\bias.
\]
The bracketed factor is invertible because $\spr((\Imat-\Dmat_W)^{-1}\bar{\mathbf A})\le\bar w<1$, and $(\Imat-\Kop)^{-1}$ is invertible by construction, so $\Lop$ is a product of invertibles, hence nonsingular. Therefore $\cvec^*=\mathbf 0\iff\bias=\mathbf 0$.

\emph{(d)} With uniform $\alpha$, $\Kop=(1-\alpha)(\Imat-\Dmat_W)^{-1}\bar{\mathbf A}+\alpha\widehat{\mathbf G}\to\widehat{\mathbf G}$ as $\alpha\to1$. The operator $\widehat{\mathbf G}$ is row-stochastic (a weighted average) with Perron eigenvalue $1$, so $(\Imat-\Kop)^{-1}$ has norm diverging as $\alpha\to1$, whence $\sup_{\lVert\bias\rVert=1}\lVert\cvec^*\rVert=\lVert\Lop\rVert\to\infty$. That $\bias=\mathbf 0$ yields $\cvec^*=\mathbf 0$ regardless of $\alpha$ is part (c).
\end{proof}

Numerically, for the worked instance of Section~\ref{sec:distortion} the contraction bound gives $\spr(\Top)\le 0.959$ against a directly computed $\spr(\Top)=0.922$; the fixed point matches long-run simulation to machine precision; and setting $\bias=\mathbf 0$ with $\alpha=0.5$ drives the concealment gap to exactly zero, confirming part (c) in its negative direction. The theorem thus delivers, rigorously, the three claims we have been drawing on. The coupled system always settles (a). When it settles, perception has fully caught up, and only concealment remains in the forcing (b). And the settled distortion is a faithful, invertible linear image of the expressive bias (c), amplified without bound by conformity but never conjured from sincerity (d).

\subsection{The regimes as parameter limits}

We can now exhibit the four regimes as limits of this one process.

\begin{proposition}[Regime limits]\label{prop:limits}
Consider the closed system~\eqref{eq:stacked}. \emph{(i) Consensus:} if $\alpha=0$, $\bias=\mathbf 0$, $\gamma_{ij}=1$, and $\Wmat$ is row-stochastic and primitive, then $\cvec^t\equiv\mathbf 0$ and $\Pmat^t\equiv\mathbf 0$, so $\hvec^t\equiv\mathbf 0$ and Proposition~\ref{prop:consensus} applies. \emph{(ii) Polarization:} under the same expression and perception limits but with $\Wmat=\Dmat\Wmat_+\Dmat$ balanced, $\hvec^t\equiv\mathbf 0$ and Proposition~\ref{prop:polarization} applies. \emph{(iii) Distortion:} if $\bias\neq\mathbf 0$, $\bar\alpha<1$, $\gamma_{\min}>0$, and $\bar w<1$, then by Theorem~\ref{thm:closure} the system converges with $\hvec_{\mathrm{perc}}^t\to\mathbf 0$ and persistent concealment $\cvec^*=\Lop\bias\neq\mathbf 0$; if instead $\bias=\mathbf 0$, the fixed point is $z^*=\mathbf 0$ and pluralistic ignorance vanishes asymptotically even with $\alpha>0$. \emph{(iv) Oscillation:} if $\Wmat$ has a unit-modulus eigenvalue other than $1$ and the forcing is asymptotically negligible, Proposition~\ref{prop:oscillation} applies; this requires leaving the hypothesis $\bar w<1$ of Theorem~\ref{thm:closure}, which is what places the regime at the spectral boundary.
\end{proposition}

\begin{proof}
Clauses (i) and (ii): $\alpha=0,\bias=\mathbf 0$ give $x_i^t=b_i^t$, hence $\cvec^t=\mathbf 0$; $\gamma=1$ gives $B_{ij}^{t+1}=x_j^t=b_j^t$, hence $\Pmat^t=\mathbf 0$; both channels vanish, so $\hvec^t\equiv\mathbf 0$ and the sign structure of $\Wmat$ (unsigned for (i), balanced for (ii)) is immaterial to the forcing, leaving the pure influence dynamics of the cited propositions. Clause (iii) is Theorem~\ref{thm:closure}(b)--(d). Clause (iv) is Proposition~\ref{prop:oscillation}, and the parenthetical observes that a unit-modulus eigenvalue violates $\spr(\Top)<1$, which under $\bar w<1$ cannot occur, so the regime lies outside the theorem's hypotheses.
\end{proof}

The practical and philosophical force of Proposition~\ref{prop:limits} is that movement between the regimes is \emph{continuous in the parameters of social epistemic life}. Raising conformity slides a consensual population toward entrenched distortion; flipping the sign structure of trust slides it toward polarization; pushing an influence eigenvalue onto the unit circle slides it toward perpetual churn. The same instrument that explains how a group comes to agree explains, at other parameter settings, how it comes to divide, to deceive itself, or to never rest. Consensus, polarization, pluralistic ignorance, and instability are not four different sociologies requiring four different theories; they are four regions of a single parameter space, and a population can be moved among them---by circumstance or by design---through changes in conformity, in the sign structure of its trust, and in the coherence of its influence. It is this last possibility, movement \emph{by design}, that Part~\ref{part:adversarial} takes up.

\part{Coupling Across Propositions}\label{part:coupling}

\section{Activating the third mode: inferential coupling}\label{sec:coupling}

Until now every operation has acted on a single proposition slice, and the third mode of the tensor---the proposition index $p$---has been a passive label. This is the proposition-separability we flagged in Section~\ref{sec:tensor}: with the propositions uncoupled, the tensor is merely a stack of independent matrices, and nothing genuinely tensorial occurs. But propositions are not independent in real epistemic life. To believe one thing commits one, on pain of incoherence, to believing others; evidence for one proposition is evidence for or against its logical and evidential neighbors; and a skilled manipulator can exploit these links, moving belief on a target proposition by moving belief on a proposition inferentially tied to it. Representing this requires letting the third mode act, and that is a genuinely multilinear operation.

Let $\Lop\in\real^{K\times K}$ be an \emph{inference operator}, with $L_{pq}$ the weight that belief in proposition $q$ lends, through logical or evidential connection, to belief in proposition $p$. Setting $\Lop=\Imat$ recovers separability (each proposition updates only from itself). The coupled belief update contracts the tensor along \emph{both} the agent mode and the proposition mode:
\begin{equation}
b_i^{t+1,p} = \sum_{q=1}^K L_{pq}\sum_{j=1}^N W_{ij}\, B_{ij}^{t,q}.
\end{equation}
Under veridical perception this is the two-sided multilinear map $\bvec^{t+1}=\Wmat\bvec^t\Lop^\top$ acting on the $N\times K$ belief matrix, and its vectorization is governed by a Kronecker product,
\begin{equation}
\vecop(\bvec^{t+1}) = (\Lop\otimes\Wmat)\,\vecop(\bvec^t),
\end{equation}
so the coupled dynamics are controlled by the spectrum of $\Lop\otimes\Wmat$, which by a standard fact \citep{horn2013} is exactly the set of products of the two spectra. This product structure has a consequence that is both mathematically clean and, we will argue, of real epistemic and adversarial significance.

\begin{proposition}[Product-spectrum law]\label{prop:product}
The eigenvalues of $\Lop\otimes\Wmat$ are the products $\mu\lambda$ over eigenvalues $\mu$ of $\Lop$ and $\lambda$ of $\Wmat$. Consequently: \emph{(i)} if both $\Lop$ and $\Wmat$ are row-stochastic and primitive, the coupled system reaches a joint consensus $(\pi_L\otimes\pi_W)^\top\vecop(\bvec^0)$ across agents and propositions, inferential spillover blending the propositionwise consensus values; \emph{(ii)} if $\spr(\Wmat)<1$ but $\spr(\Lop)\spr(\Wmat)>1$, the coupled beliefs diverge although every proposition slice is stable in isolation---mutually reinforcing inference destabilizes an otherwise-stable network; \emph{(iii)} if $\Lop=\Dmat_L\Lop_+\Dmat_L$ is balanced signed, beliefs polarize across the proposition mode, contradictory propositions acquiring mirror-image belief profiles.
\end{proposition}

\begin{proof}
The spectrum of a Kronecker product is the set of pairwise products of eigenvalues \citep[Theorem~4.2.12]{horn2013}; this is the first claim, and the eigenvectors are the Kronecker products of the factors' eigenvectors. For (i), $(\Lop\otimes\Wmat)^m=\Lop^m\otimes\Wmat^m$ is entrywise positive for $m$ large (both factors primitive), so $\Lop\otimes\Wmat$ is primitive row-stochastic and Proposition~\ref{prop:consensus} applies with left Perron vector $\pi_L\otimes\pi_W$. For (ii), $\spr(\Lop\otimes\Wmat)=\spr(\Lop)\spr(\Wmat)>1$ with the dominant eigenvector generically excited, so the iteration diverges even though $\spr(\Wmat)<1$ makes each slice stable. For (iii), $(\Dmat_L\otimes\Imat)^2=\Imat$ and the gauge argument of Proposition~\ref{prop:polarization} applies to $\Lop\otimes\Wmat=(\Dmat_L\otimes\Imat)(\Lop_+\otimes\Wmat)(\Dmat_L\otimes\Imat)$.
\end{proof}

Clause (ii) is the striking one, and it is the formal image of the greater-fool bubble with which we began. A population may be perfectly stable on every proposition considered in isolation---each issue, on its own, would settle---and yet be driven to divergence by the \emph{inferential links between the issues}, if those links are strong enough that $\spr(\Lop)\spr(\Wmat)>1$. Belief runs away not because any single issue is unstable but because believing $A$ raises belief in $B$, which raises belief in $A$, in a mutually reinforcing loop across propositions. This is issue-bundling as a destabilization mechanism, and Part~\ref{part:adversarial} will show it is available as a deliberate attack: a manipulator who cannot destabilize a population on any single proposition may be able to do so by forging inferential links between propositions, spending only the ``coupling budget'' $1-\spr(\Lop)\spr(\Wmat)$.

Finally, the closure theorem extends to the coupled setting, so that the entire three-layer analysis survives inferential coupling intact.

\begin{corollary}[Coupled closure]\label{cor:coupled}
Theorem~\ref{thm:closure} holds for the coupled three-layer system with $\bar w$ replaced by $\bar\ell\bar w$, where $\bar\ell=\max_p\sum_q\lvert L_{pq}\rvert$: since the absolute row sums of the composite contraction multiply exactly, $\lVert\Lop\otimes\Wmat\rVert_\infty=\bar\ell\bar w$, the closed system converges whenever $\bar\ell\bar w<1$, $\bar\alpha<1$, and $\gamma_{\min}>0$, and the channel-separation conclusions (b)--(d) hold unchanged.
\end{corollary}

\begin{proof}
The Kronecker product satisfies $\lVert\Lop\otimes\Wmat\rVert_\infty=\lVert\Lop\rVert_\infty\lVert\Wmat\rVert_\infty=\bar\ell\bar w$; substituting $\bar\ell\bar w$ for $\bar w$ in the contraction constant $\kappa$ of Theorem~\ref{thm:closure}(a) leaves the argument otherwise unchanged, and parts (b)--(d) never used the internal structure of the belief-transmission block beyond its norm.
\end{proof}

\paragraph{Worked example.} Four agents, with an influence matrix having $\spr(\Wmat)=0.8$---a comfortably stable slice, on its own convergent---and two propositions coupled by the symmetric inference operator
\begin{equation}
\Lop(\delta) = \begin{pmatrix} 1 & \delta\\ \delta & 1\end{pmatrix}, \qquad \spr(\Lop)=1+\delta.
\end{equation}
The product-spectrum law gives $\spr(\Lop\otimes\Wmat)=(1+\delta)(0.8)$. At weak coupling $\delta=0.15$, this is $(1.15)(0.8)=0.92<1$, and beliefs decay: the coupling is not strong enough to overcome the slice stability. But at $\delta=0.35$, the product is $(1.35)(0.8)=1.08>1$, and the identical network---same agents, same within-proposition influence, same stable slice---now \emph{diverges} (Figure~\ref{fig:coupling}). Nothing about either proposition changed; only the inferential link between them was strengthened, from $0.15$ to $0.35$, and that alone flipped the population from stable to unstable. This is the bubble mechanic in its starkest form: the instability lives entirely in the coupling between propositions, in the mutual reinforcement of beliefs across issues, and not in any issue itself. It is also, as we will see, a template for attack.

\begin{figure}[t]
\centering
\includegraphics[width=0.62\linewidth]{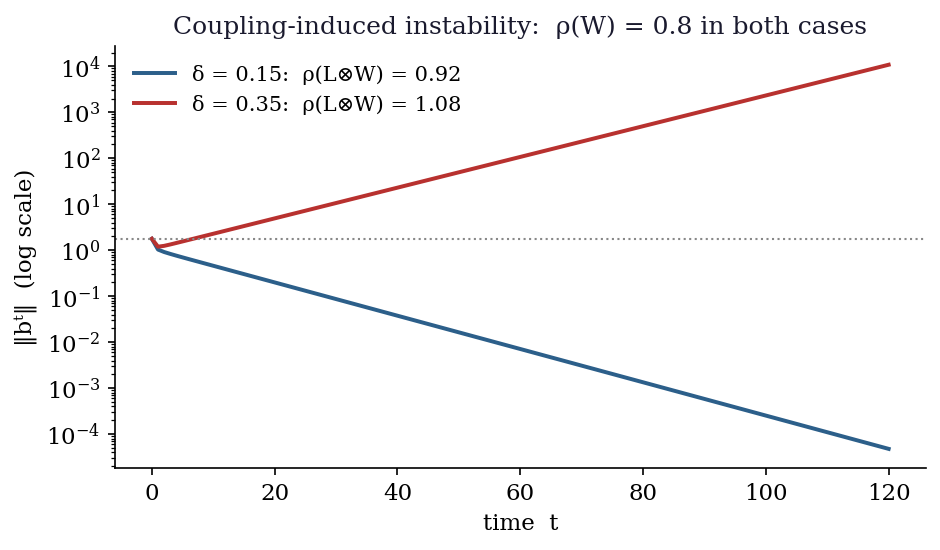}
\caption{Coupling-induced instability. With the within-proposition influence fixed at $\spr(\Wmat)=0.8$ (stable), weak inferential coupling $\delta=0.15$ gives $\spr(\Lop\otimes\Wmat)=0.92$ and decay, while stronger coupling $\delta=0.35$ gives $\spr=1.08$ and divergence. Stability is a property of the coupled operator, not of any single proposition slice.}
\label{fig:coupling}
\end{figure}

\part{The Manipulation of Social Knowledge}\label{part:adversarial}

The framework so far is descriptive: it says how interactive beliefs move, and what stable forms collective epistemic life can take. We now turn it to a normative and practical question that the descriptive theory makes newly tractable: how is social knowledge \emph{attacked}, and how can it be \emph{defended}? The contemporary urgency of the question needs little rehearsal. Influence operations, disinformation campaigns, astroturfing, and algorithmic curation are all, in one way or another, attempts to manipulate not what people believe about the world directly but what they believe others believe---the perceived climate, the apparent consensus, the sense of what is sayable. The empirical literature has documented the reach of such campaigns \citep{vosoughi2018,lazer2018,aral2019}, and a formal literature models the manipulator as a stubborn agent dragging a first-order average \citep{friedkin1990,acemoglu2011,yildiz2013}. But the characteristic modern attacks are irreducibly second-order, and a first-order model cannot represent them: a lie about what others believe has no state variable to land on. Our three-layer framework supplies exactly the missing variables, and with them an exact theory of second-order manipulation.

\section{The attack surface: three layers, three families}\label{sec:surface}

The organizing insight is that attacks on social knowledge are individuated not by the \emph{kind of move} the manipulator makes but by the \emph{epistemic layer} the move targets. The three-layer ontology of Part~\ref{part:foundations}---content, expressed climate, observation medium---induces a threefold taxonomy of attack, and the classical families of epistemic misconduct map onto the layers exactly.

\paragraph{Inferential obstructions target the content layer.} Equivocation, misdirection, planted contradiction, the manufacture of false evidence: these operate on beliefs and the inferential links between them. They attack $\bvec$ directly, or the inference operator $\Lop$ of Part~\ref{part:coupling}. Classical propaganda---a false claim about the world---is the pure case. These are \emph{first-order} attacks: they concern what is true.

\paragraph{Dramaturgical exploitations target the perceived climate.} Staged consensus, manufactured majorities, performed conviction, the ``everyone now agrees that\dots'' that no one has verified: these operate on the perception and expression layers, $\Mmat$ and $\xvec$. ``Most people think $P$'' is an injection into the perceived climate $\widehat c$; astroturfing fabricates entries of $\xvec$ corresponding to no private belief. These are \emph{second-order} attacks: they concern what is believed to be believed.

\paragraph{Context manipulations target the observation medium.} Control of the feed, the ranking, the translation, the venue, the sampling of who is heard: these corrupt the map by which one layer observes another. The perception update comes to track not $x_j^t$ but $x_j^t+a_{ij}$, where the corruption $a_{ij}$ is set by whoever mediates $i$'s observation of $j$. The manipulator here touches no one's belief and no one's expression; it owns the channel. These too are second-order, but---as we will prove---of a distinct and more durable kind than dramaturgical exploitation.

First-order models see only the content layer; everything second-order, climate and medium, is invisible to them, and it is exactly there that contemporary influence operations work. The three families are priced differently because they attack different layers, and the rest of Part~\ref{part:adversarial} is the pricing. An adversary, formally, is an actor outside the population who adds terms to the dynamics or alters its parameters, under a budget, to change where beliefs go; the state space of Part~\ref{part:foundations} makes the possible targets exhaustive---states, the channels between them, and the operators governing them---and we take the families in turn.

\section{Robustness I: attacks on states are transient}\label{sec:transience}

The most basic robustness fact is that the framework's global stability, established in the closure theorem, is \emph{itself} a security guarantee against a large class of attacks: any one-shot perturbation of the state, however large, decays.

\begin{theorem}[Transience of state attacks]\label{thm:transience}
Let the population satisfy the hypotheses of Theorem~\ref{thm:closure}, and let an adversary apply an arbitrary one-shot perturbation $\Delta z$ to the state at time $t_0$---any combination of injected beliefs, fabricated expressions, and planted perceptions, of any magnitude. Then the attacked trajectory converges to the same fixed point as the unattacked one, and the entire effect of the attack decays geometrically: the difference at time $t_0+k$ is $\Top^k\Delta z$, of norm $O(\spr(\Top)^k)$.
\end{theorem}

\begin{proof}
The dynamics are affine with fixed operator $\Top$. If $z^t$ and $\tilde z^t$ are trajectories with $\tilde z^{t_0}=z^{t_0}+\Delta z$, their difference obeys the homogeneous recursion $\tilde z^{t_0+k}-z^{t_0+k}=\Top^k\Delta z$, which tends to zero geometrically since $\spr(\Top)<1$ by Theorem~\ref{thm:closure}(a). Both trajectories share the limit $(\Imat-\Top)^{-1}\tilde\bias$.
\end{proof}

This is the population-level form of epistemic self-healing. A burst of ``most people think $P$,'' however loud, is a one-shot perturbation of the perception layer, and the perception trackers subsequently wash it out against actual expressed belief; the lie decays at the spectral rate (Figure~\ref{fig:transient}, blue). Three qualifications set the agenda for the rest of Part~\ref{part:adversarial}, and each corresponds to escaping a hypothesis of the theorem. The decay rate is governed by the spectral gap, so populations near criticality forget slowly. The theorem concerns \emph{one-shot} attacks; persistent injection is a different matter (Section~\ref{sec:persuader}). And the theorem presumes the observation medium is honest---exactly what a context manipulation revokes (Section~\ref{sec:channel}).

\begin{figure}[t]
\centering
\includegraphics[width=0.62\linewidth]{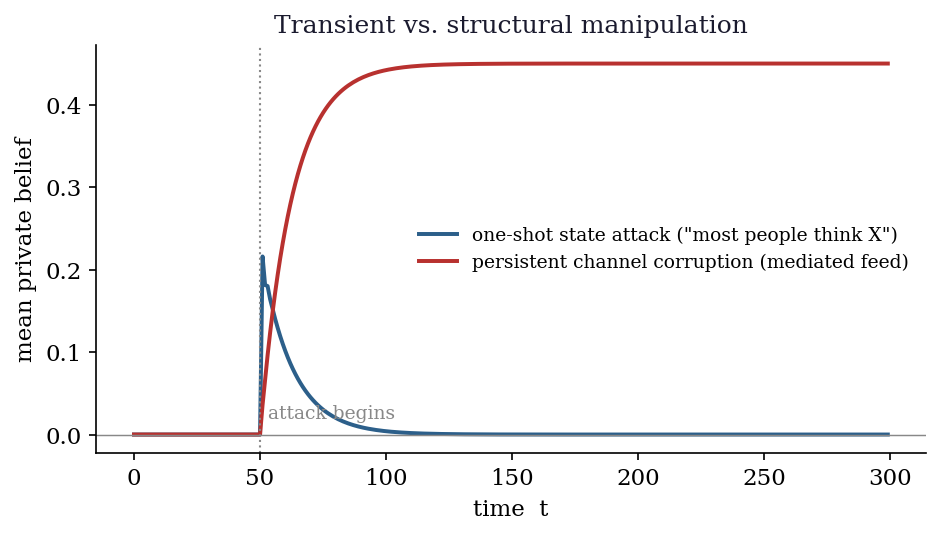}
\caption{Transient versus structural manipulation. A one-shot lie about the climate of opinion (blue) decays geometrically, per Theorem~\ref{thm:transience}; persistent corruption of the observation medium (red) displaces beliefs permanently, without ever touching any agent's belief or expression (Theorem~\ref{thm:channel}).}
\label{fig:transient}
\end{figure}

\section{The Persuader: the cost of persuasion and a ranking of the families}\label{sec:persuader}

By Theorem~\ref{thm:transience}, durable persuasion requires \emph{persistence}: the adversary must inject every period, becoming part of the environment. A persistent injection is indistinguishable from a forged component of the affine drive $\tilde\bias$, and its effect is read off the fixed point $z^*=(\Imat-\Top)^{-1}\tilde\bias$. Two injection layers matter, corresponding to two of our three families. A \emph{belief-layer} injection (a constant $\mathbf v$ added to the belief update) is an inferential obstruction, a first-order lie about the world. An \emph{expression-layer} injection (a constant $\mathbf u$ added to the expression update) is a dramaturgical exploitation, a second-order lie about what is being said. The steady-state displacements are linear in the injection, and on the symmetric influence family the gains admit exact closed forms that---remarkably---rank the two families by a single organizational parameter.

\begin{theorem}[Cost of persuasion; conformity crossover]\label{thm:crossover}
For uniform conformity $\alpha<1$ and the symmetric influence family $\Wmat=w_d\Imat+w_o(\mathbf J-\Imat)$ with $r=(N-1)w_o/(1-w_d)<1$, the steady-state gains per unit uniform injection are
\begin{equation}
g_x(\alpha) = \frac{r}{(1-\alpha)(1-r)} \quad\text{(dramaturgical, expression layer)}, \qquad
g_b = \frac{1}{(1-w_d)(1-r)} \quad\text{(inferential, belief layer)},
\label{eq:gains}
\end{equation}
so the belief-layer gain is independent of conformity while the expression-layer gain diverges as $\alpha\to1$. The two families cross in efficiency at exactly
\begin{equation}
\alpha^* = 1 - r(1-w_d).
\label{eq:crossover}
\end{equation}
Below $\alpha^*$ inferential obstruction is cheaper; above $\alpha^*$ dramaturgical exploitation strictly dominates, with advantage $g_x/g_b=r(1-w_d)/(1-\alpha)$ growing without bound.
\end{theorem}

\begin{proof}
On the uniform (consensus) direction $\one$, the normalized climate operator fixes $\one$ ($\widehat{\mathbf G}\one=\one$) and the belief-transmission contraction scales it, $(\Imat-\Dmat_W)^{-1}\bar{\mathbf A}\one=r\one$. For the expression layer, the fixed-point operator from Theorem~\ref{thm:closure}(c) satisfies $\Kop\one=[(1-\alpha)r+\alpha]\one$, so $(\Imat-\Kop)\one=(1-\alpha)(1-r)\one$, and a unit uniform expression injection $u\one$ produces $\Delta\bvec^*=(\Imat-\Dmat_W)^{-1}\bar{\mathbf A}(\Imat-\Kop)^{-1}u\one=\frac{r}{(1-\alpha)(1-r)}u\one$, giving $g_x$. For the belief layer, on the uniform direction the expression fixed-point equation gives $(1-\alpha)x=(1-\alpha)b$, hence $x=b$ for $\alpha<1$, and the belief equation becomes $b=rb+v/(1-w_d)$, so $\Delta b=v/[(1-w_d)(1-r)]$, independent of $\alpha$, giving $g_b$. Equating $g_x=g_b$ and solving for $\alpha$ yields~\eqref{eq:crossover}. All identities verified numerically to machine precision against both the closed forms and long-run simulation.
\end{proof}

Theorem~\ref{thm:crossover} does more than price two instruments; it \emph{ranks the three families of attack} and shows the ranking is governed by a single, measurable property of the target population. Because context manipulations act through the same second-order machinery as dramaturgical exploitations---both enter through perception and expression rather than through belief---the crossover sorts the families cleanly. Below $\alpha^*$, first-order attacks on content are the efficient choice; above $\alpha^*$, the two second-order families, on climate and on medium, dominate, and their advantage grows without bound as conformity rises. This yields a prediction about institutions that, to our knowledge, has not previously been stated.

\begin{corollary}[Organizational vulnerability]\label{cor:orgvuln}
Order populations by their conformity pressure $\alpha$. Then the efficient attack surface shifts monotonically with $\alpha$: low-conformity populations are most cheaply attacked through their content (inferential obstruction), high-conformity populations through their climate and medium (dramaturgical exploitation and context manipulation). At $\alpha>\alpha^*$ the second-order families strictly dominate, with efficiency advantage $\propto(1-\alpha)^{-1}$.
\end{corollary}

The corollary is a testable claim about the differential vulnerability of institutions to different kinds of epistemic attack. A high-conformity organization---one whose members resolve a large share of their expressed positions by reference to the perceived consensus rather than to private conviction---is systematically \emph{more} vulnerable to staged consensus and to control of its internal communication channels than to the planting of false facts, and the gap widens the more conformist it is. The reverse holds for low-conformity, high-dissent cultures, where inferential obstruction remains the efficient attack because there is little perceived-consensus machinery to exploit. For the worked instance of Section~\ref{sec:distortion} ($r\approx0.474$, $w_d=0.24$), the crossover sits at $\alpha^*=0.64$: a body that resolves two-thirds of its expression by perceived consensus has already made second-order attack the efficient instrument, and at $\alpha=0.9$ the dramaturgical advantage is $3.6$-fold (Figure~\ref{fig:crossover}). An adversary optimizing under a budget will migrate from content to climate and medium precisely as a population's conformity rises---and the migration is rational, a response to the changing relative price of the two kinds of lie, not a matter of taste.

\begin{figure}[t]
\centering
\includegraphics[width=0.62\linewidth]{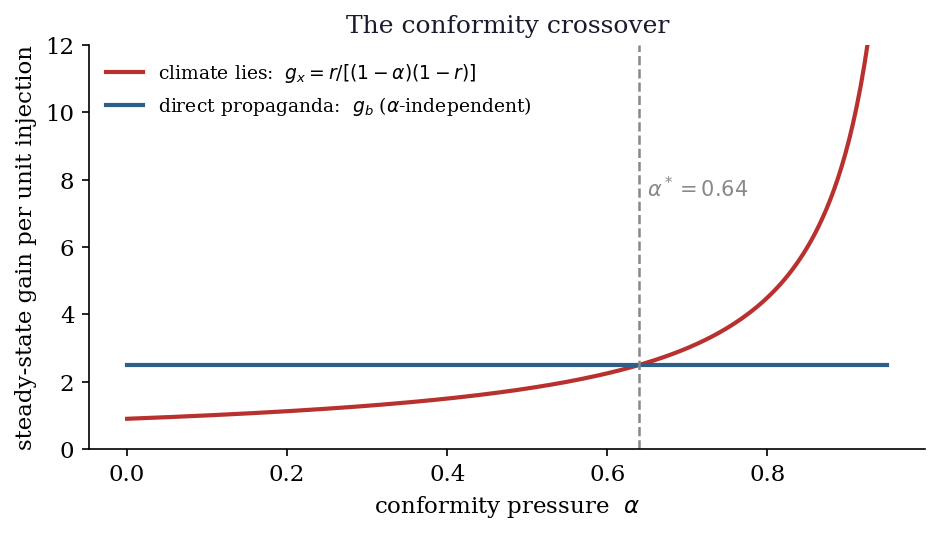}
\caption{The conformity crossover ranks the families. The gain of dramaturgical exploitation (climate) grows as $1/(1-\alpha)$ while the gain of inferential obstruction (content) is $\alpha$-independent; above $\alpha^*$ the second-order families---dramaturgical and, through the same machinery, contextual---strictly dominate, with advantage growing without bound.}
\label{fig:crossover}
\end{figure}

Automated amplification fits the same schedule. A bot is a pure-expression agent: it has no private belief, only a fixed profession it repeats.

\begin{proposition}[Bots are structured persistent injections]\label{prop:bots}
Add to the population an agent with no belief dynamics and fixed expression $s$, receiving influence weights $W_{i,\mathrm{bot}}$ from humans. The human block of the resulting fixed point equals exactly the fixed point of the human-only system with persistent belief-layer injection $v_i=W_{i,\mathrm{bot}}s$ and climate injection $\lvert W_{i,\mathrm{bot}}\rvert s/\sum_j\lvert W_{ij}\rvert$. Automated amplification therefore inherits the price schedule of Theorem~\ref{thm:crossover}, with displacement bounded by the bot's influence share; and a bot is structurally identifiable as an agent whose concealment gap is undefined---an expression anchored to no belief.
\end{proposition}

\begin{proof}
At any fixed point the bot's trackers satisfy $M_{i,\mathrm{bot}}^*=s$ (its expression is constant), so its contributions to $i$'s belief aggregation and perceived climate are the stated constants; block-eliminating the bot's coordinates from the stacked fixed-point system leaves the human-only system with those constants as injections. Verified numerically against full simulation of the extended system.
\end{proof}

\section{Context manipulation: when robustness is revoked}\label{sec:channel}

The self-healing of Theorem~\ref{thm:transience} rests on an assumption so natural it is nearly invisible: that agents observe one another's actual expressions. Whoever mediates that observation---a feed, a ranking algorithm, an editor, a translator, a pollster who chooses whom to sample---can revoke it. This is the third family, context manipulation, and it behaves categorically differently from the other two.

\begin{theorem}[Context manipulation produces permanent displacement]\label{thm:channel}
Suppose $i$'s perception of $j$ tracks $x_j^t+a_{ij}$ for a fixed corruption matrix $\Amat$ (zero diagonal). Then under the hypotheses of Theorem~\ref{thm:closure} the system still converges, but its fixed point satisfies $M_{ij}^*=x_j^*+a_{ij}$: the perception-lag matrix no longer vanishes but equals the corruption exactly, $\Pmat^*=\Amat$, and beliefs are permanently displaced by a nonsingular linear image of $\Amat$---even with fully sincere agents ($\bias=\mathbf 0$):
\begin{equation}
\xvec^*=(\Imat-\Kop)^{-1}\bigl[(1-\alpha)(\Imat-\Dmat_W)^{-1}(\Wmat\odot\Amat)\one+\alpha\,\widehat g_A\bigr], \qquad
\bvec^*=(\Imat-\Dmat_W)^{-1}\bigl[\bar{\mathbf A}\xvec^*+(\Wmat\odot\Amat)\one\bigr],
\end{equation}
where $\widehat g_A$ is the climate-weighted image of the corruption.
\end{theorem}

\begin{proof}
The corrupted perception row reads $M_{ij}^*=(1-\gamma_{ij})M_{ij}^*+\gamma_{ij}(x_j^*+a_{ij})$ at the fixed point, forcing $M_{ij}^*=x_j^*+a_{ij}$ for $\gamma_{ij}>0$; hence $\Pmat^*=[M_{ij}^*-x_j^*]=\Amat$. Substituting $\Mmat^*=\mathbf S\xvec^*+\Amat$ into the belief and expression rows yields the stated linear system; convergence is unaffected because the corruption enters only the constant (affine) term, leaving $\Top$ and hence $\spr(\Top)<1$ unchanged. Verified numerically to machine precision, including $\Pmat^*=\Amat$ exactly.
\end{proof}

The contrast with Theorem~\ref{thm:transience} is the deepest division in the theory of epistemic attack, and it cuts the phenomenon of collective misbelief into two kinds that are routinely conflated and that require opposite remedies.

\paragraph{Episodic propagation is tied to discovery and is recoverable.} A false claim injected into beliefs or expressions---an inferential obstruction or a dramaturgical exploitation, applied once or even repeatedly but as discrete events---is a \emph{state} attack. By Theorem~\ref{thm:transience} its effect decays once the injection stops, and the natural remedy is \emph{exposure}: revealing the falsehood removes the forcing, and the population relaxes to its uncorrupted fixed point. Discovery is curative because the distortion lives in the state, and the state heals; the recovery timescale is set by the spectral gap. This is the kind of misinformation that fact-checking addresses, and against which fact-checking is, in principle, sufficient.

\paragraph{Structural sedimentation is architecture-level and impervious to disclosure.} A context manipulation---corruption of the medium through which agents observe one another---is not a state attack but a boundary condition. By Theorem~\ref{thm:channel} it displaces the fixed point \emph{permanently}, and, crucially, \emph{regardless of whether the manipulation is discovered}. Exposing that a feed is ranked, or that a translation is slanted, or that a poll oversamples, does nothing on its own: as long as the corrupted medium remains in place, $\Pmat^*=\Amat$ and the displacement persists. The distortion has sedimented into the architecture; it is a property of the channel, not of any act transmitted through it, and no quantity of revelation about the acts removes it. The only remedy is architectural---restoring an honest medium, which restores the vanishing of the perception term (Theorem~\ref{thm:closure}(b)) and lets the residual displacement decay by Theorem~\ref{thm:transience}.

The practical error this distinction guards against is applying the remedy for one kind to the other, and it is a consequential error. Treat structural sedimentation as though it were episodic propagation---answer a corrupted recommender with fact-checks and disclosures while leaving the recommender in place---and one addresses the class of attack the architecture is \emph{not} carrying. Matters can then get worse. Each debunked episode is regenerated by the standing architecture, and effort spent on exposure is effort not spent on the only intervention that works, which is to change the channel. Treat episodic propagation as structural, in the other direction---re-architect a communication channel in panicked response to a transient rumor that Theorem~\ref{thm:transience} guarantees would have decayed on its own---and one over-corrects at cost, risking damage to an epistemic infrastructure that was functioning. A rumor is an injection. A recommender is a boundary condition. They are different mathematical objects, and they do not admit the same cure.

\section{The Concealer: ambient distortion as a distinct equilibrium}\label{sec:concealer}

The third adversary attacks neither states nor structure but the \emph{expression environment} itself: it raises the price of sincerity. Pile-ons, public example-making, selective enforcement, the strategic deployment of outrage: these act as an induced expressive bias $\bias\neq\mathbf 0$ (and exert upward pressure on conformity $\alpha$). By the closure theorem, the population then converges to genuine, self-sustaining pluralistic ignorance, $\cvec^*=\Lop\bias$, amplified without bound as conformity rises. Nothing in the attack need be false: every individual act may be a truthful report of a real punishment for dissent. The distortion is manufactured entirely out of incentives, not out of falsehoods, which is what makes it invisible to any defense aimed at falsehood.

This gives a precise mechanism to a distinction often drawn descriptively but rarely explained: the difference between \emph{episodic} mistrust---particular breaches, accumulating---and \emph{ambient} mistrust, the diffuse standing sense that others cannot be trusted or that candor is unsafe. The framework shows these are not the same phenomenon at different intensities. Episodic mistrust is a sequence of state attacks: discrete injections that, by Theorem~\ref{thm:transience}, decay unless renewed, so that ambient mistrust construed as accumulated episodes would require continual replenishment and would relax the moment the episodes ceased. Ambient mistrust in the present theory is something else entirely: a \emph{distinct equilibrium} $\cvec^*=\Lop\bias$, produced not by any episode but by the standing expression environment---the cost of sincerity encoded in $\bias$. It does not accumulate from episodes and does not decay when episodes stop; it is a fixed point, sustained by architecture rather than by events. Two populations with identical episodic histories can sit at different ambient equilibria if their expression environments differ, and a single population's ambient level is set by $\bias$ and $\alpha$, not by its recent breaches. The distinction matters because the two call for different responses, and---as with the episodic/structural divide of the previous section---applying the wrong one is worse than doing nothing.

\begin{corollary}[Transparency is inert against, and can entrench, ambient distortion]\label{cor:transparency}
Let the concealment equilibrium be $\cvec^*=\Lop\bias$ with $\bias\neq\mathbf 0$. An intervention that increases the observability of expressed belief without reducing the cost of sincere dissent leaves $\bias$ unchanged, hence leaves $\cvec^*$ unchanged; and if by making endorsement more visible it raises the perceived-consensus weight and thus effective conformity $\alpha$, it \emph{increases} $\lVert\cvec^*\rVert$, deepening the distortion it was meant to cure.
\end{corollary}

\begin{proof}
By Theorem~\ref{thm:closure}(b)--(c) the concealment fixed point depends on the expression environment through $\bias$ and, via $\Lop$ and $\Kop$, through $\alpha$; it is independent of the observability of expression, since the perception channel vanishes at the fixed point and the concealment channel does not involve how well expression is perceived. Raising observability of endorsement without altering the sincerity cost holds $\bias$ fixed, leaving $\cvec^*=\Lop\bias$ fixed. If in addition it raises $\alpha$, then since $\lVert\Lop\rVert$ is increasing in $\alpha$ with $\lVert\Lop\rVert\to\infty$ as $\alpha\to1$ by Theorem~\ref{thm:closure}(d), $\lVert\cvec^*\rVert=\lVert\Lop\bias\rVert$ increases.
\end{proof}

This is why transparency initiatives so often fail against ambient distortion, and it says precisely how they fail. Making expressed endorsement more visible---publishing who supported what, surfacing the majority position, broadcasting rates of compliance---does not touch the quantity that produces pluralistic ignorance. That quantity is the gap between private belief and public expression, driven by the cost of dissent. Worse, making the ambient endorsement more salient can raise the very conformity that amplifies the gap. A well-intentioned transparency program then \emph{reinforces} the concealment equilibrium rather than dissolving it: agents see more clearly what everyone is expressing, infer the consensus is even stronger than they thought, and conceal their dissent even more. The only interventions that move $\cvec^*$ are those that lower the cost of sincere dissent---protected disagreement, anonymity where appropriate, the secret ballot---because those alone reduce $\bias$. Against the Persuader and the context manipulator, exposure and architectural repair respectively help. Against the Concealer, only the price of sincerity matters, and transparency aimed at the wrong quantity is not merely inert but counterproductive. This is perhaps the sharpest practical lesson of the adversarial theory. The instinct to answer collective misperception with more visibility is exactly wrong when the misperception is a concealment equilibrium, and the framework says when that is the case.

\section{Backchannels, trust conduits, and the closed epistemic walk}\label{sec:gossip}

``I heard $X$ say that $Y$ thinks $P$'' is an injection into a composed walk of the attribution slice, in the sense of Theorem~\ref{thm:path}. Two phenomena usually treated separately---the \emph{backchannel} (a rumor passed along a chain of intermediaries) and the \emph{trust conduit} (corroboration that returns to its source through the network)---are, in the walk calculus, the same object seen from two ends: both are \emph{closed epistemic walks}, sequences of attributions that begin and end at the same node. Theorem~\ref{thm:path} prices them, and the pricing yields a concrete audit rule.

A planted attribution reaching a listener through a chain of length $L$ is weighted by the product of the $L$ attribution weights along the chain; when $\spr(\Bten)<1$ these products decay geometrically in depth (Theorem~\ref{thm:closed}), so \emph{deep backchannel gossip is self-attenuating}---a third-hand attribution is structurally cheaper to plant but structurally weaker on arrival. The danger is not depth but closure.

\begin{theorem}[Echo multiplier of a closed walk]\label{thm:echo}
Let $\spr(\Bten)<1$. An attribution planted on the directed edge $(X,Y)$ returns to its planting site through walks from $Y$ back to $X$ of every length, with total weight
\begin{equation}
\mathcal{E}_{XY} = \sum_{L\ge0}(\Bten^L)_{YX} = \bigl[(\Imat-\Bten)^{-1}\bigr]_{YX},
\label{eq:echo}
\end{equation}
the \emph{echo multiplier}; and the sensitivity of the population's total closed-walk mass to the planted edge is the squared resolvent, $\partial\tr(\Imat-\Bten)^{-1}/\partial B_{XY}=[(\Imat-\Bten)^{-2}]_{YX}$. For nonnegative primitive attribution weights, every entry of the resolvent grows without bound as $\spr(\Bten)\to1$: as the attribution operator approaches spectral radius one, the resolvent develops a pole, and a single planted attribution generates unbounded apparent corroboration---the claim returns to its planter, and to everyone else, as seemingly independent confirmation. An echo chamber, in the quantitative sense the walk calculus makes available (Section~\ref{sec:invariance}), is an attribution network whose spectral radius approaches one: the divergence is a resolvent pole---a property of accumulated attribution weight, not of common belief---and it is a pole, not a critical point, carrying no scaling exponent.
\end{theorem}

\begin{proof}
The Neumann series $\sum_{L\ge0}\Bten^L=(\Imat-\Bten)^{-1}$ converges since $\spr(\Bten)<1$, and its $(Y,X)$ entry sums the weights of all walks from $Y$ to $X$ of every length, by Theorem~\ref{thm:path}. The trace derivative is $\partial\tr(\Imat-\Bten)^{-1}/\partial B_{XY}=\tr[(\Imat-\Bten)^{-1}\mathbf e_X\mathbf e_Y^\top(\Imat-\Bten)^{-1}]=[(\Imat-\Bten)^{-2}]_{YX}$, using $\partial(\Imat-\Bten)^{-1}=(\Imat-\Bten)^{-1}(\partial\Bten)(\Imat-\Bten)^{-1}$. Divergence as the Perron root approaches $1$ follows from Perron--Frobenius applied to $\sum_L\Bten^L$. Both identities verified numerically to machine precision.
\end{proof}

The result gives the backchannel and the trust conduit a common formal anchor and a single defensive rule (Figure~\ref{fig:echo}). Because both are closed walks, both are priced by~\eqref{eq:echo}, and both are dangerous exactly when the attribution network runs near $\spr(\Bten)=1$. The audit rule follows from the structure of the walk:
\begin{quote}
\emph{Demand the belief at the end of the attribution chain; treat closed chains as structurally adversarial.}
\end{quote}
An attribution chain that terminates in a private belief---a $b$ at the end of the walk, someone who actually holds the view rather than merely reporting that others hold it---is an open walk anchored in ground truth, and is genuine evidence. A chain that closes on itself---``people are saying that people are saying,'' corroboration with no terminating believer---is structurally an echo, and by Theorem~\ref{thm:echo} its apparent evidential weight is manufactured by the network's own recirculation, not supplied by any independent source. The rule operationalizes the theorem: trace each attribution to its terminal, and if the terminal is not a belief but a return to the chain's own origin, discount the corroboration to zero. Closed chains are presumptively adversarial not because they are always deliberately planted but because their apparent evidential weight is generated by closure itself, which an adversary operating a network near spectral criticality can exploit at will.

\begin{figure}[t]
\centering
\includegraphics[width=0.62\linewidth]{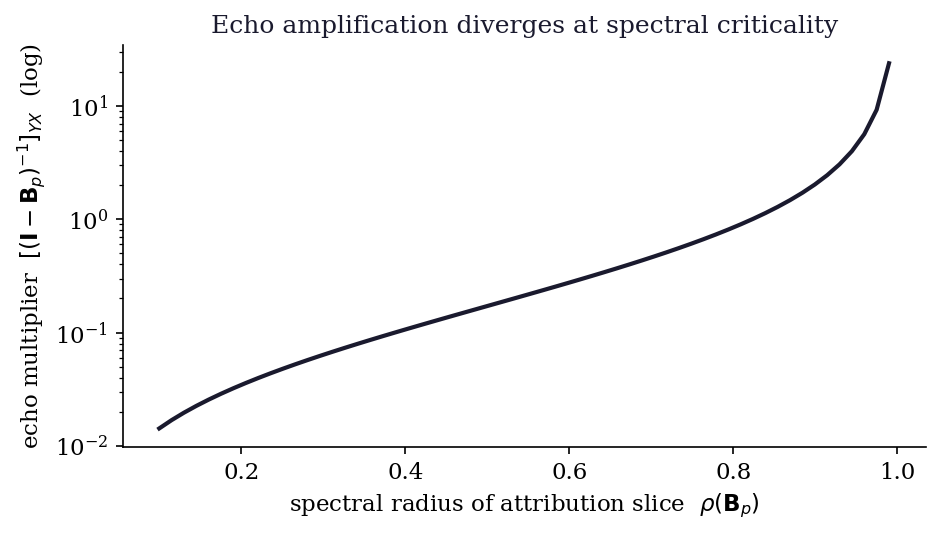}
\caption{The echo multiplier is a resolvent entry~\eqref{eq:echo} and grows without bound as the attribution network approaches spectral radius one---a resolvent pole. Near $\spr=1$, a single planted attribution returns as unbounded apparent corroboration: the formal signature of an echo chamber.}
\label{fig:echo}
\end{figure}

Finally, Assumption~\ref{ass:composition} earns its keep defensively, as anticipated in Section~\ref{sec:higher}: compositional cognition compresses the attack surface available to an adversary who injects higher-order attributions.

\begin{proposition}[Compositional cognition compresses the attack surface]\label{prop:composition}
An agent who stores belief hierarchies as free parameters exposes an attack surface of $N^L K$ independently injectable attributions at depth $L$. An agent satisfying Assumption~\ref{ass:composition} stores $N^2K$ parameters in total, from which all higher-order attributions are computed; any injected higher-order claim inconsistent with the agent's own composed estimate $(\Bten^2)_{iY}$ is detectable by direct comparison. Compositionality does not prevent first-order deception, but it makes higher-order deception either redundant (if consistent with the composition) or detectable (if not).
\end{proposition}

\begin{proof}
The parameter counts are immediate: a free depth-$L$ hierarchy has $N^L$ attribution chains per proposition, while a compositional agent stores only the $N^2$ first-order attributions per proposition and computes the rest. For detectability, a compositional agent's second-order attribution is determined by its first-order state, so an asserted value $\tilde B$ for ``$X$'s estimate of $Y$'' admits the consistency test $\tilde B\overset{?}{=}(\Bten^2)_{iY}$, computable from the agent's own $N^2$ entries; rejection flags the injection.
\end{proof}

\section{Defense by design}\label{sec:defense}

Each family meets its counter in a parameter, and because the families are individuated by the layer they target, the counters are as distinct as the targets. Against inferential obstruction (content layer): \emph{exposure}. First-order lies are state attacks; by Theorem~\ref{thm:transience} they decay once revealed, so fact-checking is the correct and sufficient remedy \emph{for this family alone}. Against dramaturgical exploitation (climate layer): \emph{lower conformity}. Theorem~\ref{thm:crossover} makes conformity the inverse price of climate manipulation, so institutions of sincere expression are quantitative defenses; the secret ballot is an institutionalized $\alpha\approx0$ reading of private belief, which is exactly why elections ``surprise'' manipulated public discourses---the ballot reads $\bvec$ while the discourse reads $\xvec$. Against context manipulation (medium layer): \emph{architecture, not disclosure}. Theorem~\ref{thm:channel} locates permanent distortion in the medium, and Section~\ref{sec:channel} shows exposure is inert against it; the defense is verifiable, unranked access to what others actually said. Against the Perverter (structure): \emph{spectral slack}. The margins $1-\bar w$ and $1-\spr(\Lop)\spr(\Wmat)$ are stability reserves; a polity can afford asserted inferential entailments only in proportion to its slack. Against the Concealer (expression environment): \emph{lower the cost of sincerity}, since by Corollary~\ref{cor:transparency} transparency is inert or counterproductive here. Against closed-walk corroboration: \emph{audit chains to their terminals} and operate the attribution network with spectral slack, where the echo multiplier is bounded.

The composite lesson inverts the emphasis of content-centric defense. The three families target three layers, and the single most common error in practice is to apply the content-layer remedy---exposure---to a climate-layer or medium-layer attack, where it is respectively insufficient or actively counterproductive. Moderating messages addresses inferential obstruction, which is the one family the system already heals on its own; it does nothing for the dramaturgical and contextual attacks that dominate precisely in the high-conformity organizations most tempted to reach for it. A theory of the defense of social knowledge must begin by identifying which layer is under attack, because the remedy that repairs one layer can deepen the damage to another.

\part{The Transient Fragility of Hierarchies}\label{part:physics}

\section{Non-normal reactivity: why spectral safety is not transient safety}\label{sec:reactivity}

The robustness theorem of Part~\ref{part:adversarial} (Theorem~\ref{thm:transience}) certifies that one-shot attacks decay, because $\spr(\Top)<1$. But asymptotic decay is a statement about the long run, and it conceals a phenomenon that the short run can hide: a perturbation may grow, sometimes enormously, \emph{before} it decays. Whether it does depends on a property of the operator that the spectrum does not see---its \emph{non-normality}---and the property is generic in exactly the epistemic structures that describe real institutions: directed, hierarchical ones. This final analytical part establishes that hierarchical epistemic populations are \emph{reactive}---provably amplifying, at arbitrarily small spectral radius, perturbations that the spectral analysis certifies as safe---and reports numerically that the transient gain grows with hierarchy depth over the accessible range, while being candid about the limits of that scaling claim. It is the deepest layer of the fragility story, and it revises the robustness intuition with which the study of these systems usually begins.

\subsection{Reactivity and the Kreiss constant}

An operator $\Top$ is \emph{normal} if it commutes with its transpose, $\Top\Top^\top=\Top^\top\Top$, equivalently if its eigenvectors can be chosen orthogonal. For a normal operator the spectrum controls everything: $\lVert\Top^k\rVert=\spr(\Top)^k$, and a perturbation decays monotonically at the spectral rate. Real epistemic influence, however, is directed and asymmetric---leaders are heard without listening, information flows down chains of command, attention is unreciprocated---and directed influence matrices are generically \emph{non-normal}. For a non-normal operator the eigenvectors are far from orthogonal, and a perturbation expressed in that skewed basis can have its components partially cancel initially and then, as they decay at different rates and rotate, transiently reinforce, producing growth before the eventual spectral decay asserts itself. This is \emph{reactivity} \citep{neubert1997,trefethen2005,asllani2018}, and it is the manipulator's opening: an attack that the spectrum certifies as transient can be amplified by orders of magnitude on its way to decaying, and repeated at the timescale of maximum amplification.

The relevant quantity is not the spectral radius but the \emph{transient gain}
\begin{equation}
\Gamma = \sup_{k\ge0}\lVert\Top^k\rVert_2,
\end{equation}
the largest amplification any perturbation attains over all horizons. It is bracketed by the \emph{Kreiss constant}
\begin{equation}
\Kreiss(\Top) = \sup_{\lvert \zeta\rvert>1}(\lvert\zeta\rvert-1)\,\lVert(\zeta\Imat-\Top)^{-1}\rVert_2
\end{equation}
through the Kreiss matrix theorem, $\Kreiss(\Top)\le\Gamma\le e\,n\,\Kreiss(\Top)$, and it is diagnosed geometrically by the protrusion of the $\varepsilon$-pseudospectrum---the set of $\zeta$ where $\lVert(\zeta\Imat-\Top)^{-1}\rVert_2\ge\varepsilon^{-1}$---beyond the unit circle. For a normal operator $\Gamma=1$ and the Kreiss constant is $1$: no transient growth is possible, and the spectral certificate is exact. We confirmed this for symmetric, undirected interactive-belief populations, where $\Top$ is nearly normal and $\Gamma\approx1.03$ (Figure~\ref{fig:pseudospectra}, left): the eigenvalues and the pseudospectral contours nearly coincide, and Theorem~\ref{thm:transience}'s reassurance is quantitatively exact. Hierarchy destroys that reassurance.

\subsection{The square-root reactivity law}

Consider a feed-forward epistemic hierarchy: agent $i$ attends only to agents above it in a ranking, and no one attends downward (opinion leaders are heard but do not listen), with fast perception $\gamma\to1$. The influence matrix is then nearly nilpotent, so its spectral radius---and that of the stacked operator $\Top$---can be driven arbitrarily small. Yet the pseudospectrum bulges dramatically toward and past the unit circle (Figure~\ref{fig:pseudospectra}, right): with $\spr(\Top)=0.10$, at $\varepsilon=0.05$ the pseudospectral radius exceeds the spectral radius by $0.85$---a seventeen-fold larger excursion beyond the spectrum than in the symmetric case. The eigenvalues cluster near the origin while the resolvent norm is large out near the unit circle, the defining geometric signature of non-normality.

\begin{figure}[t]
\centering
\includegraphics[width=0.95\linewidth]{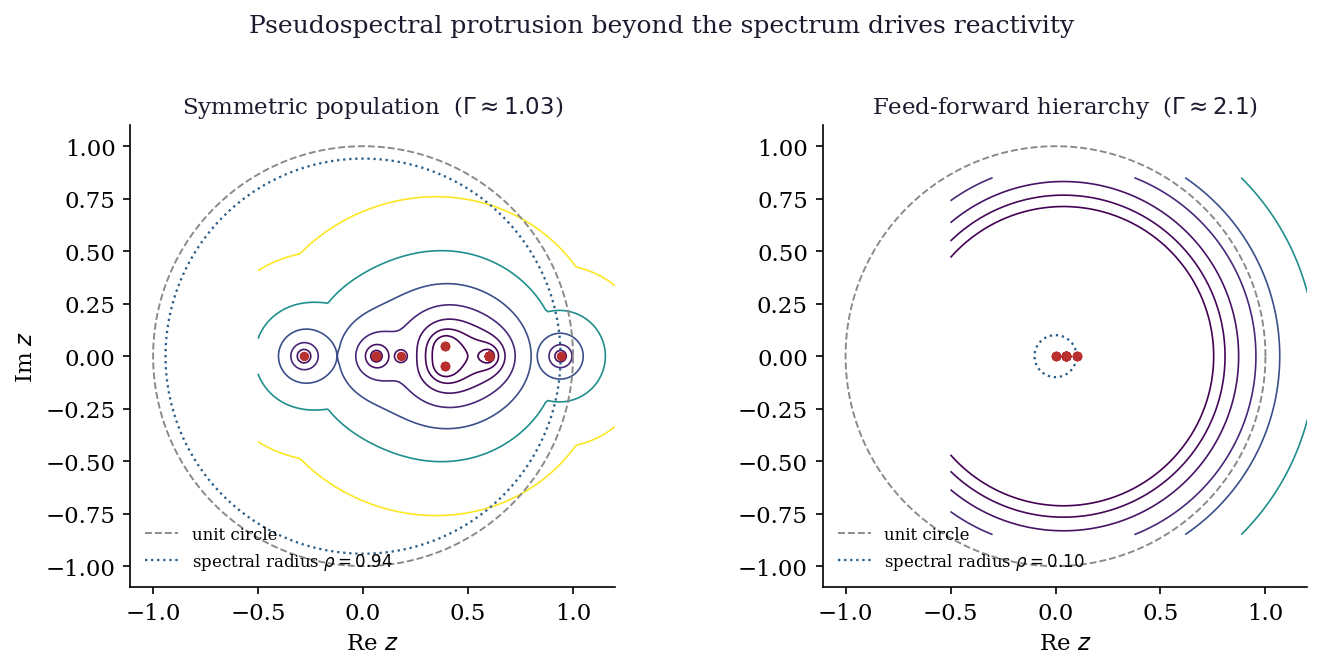}
\caption{$\varepsilon$-pseudospectra of the stacked interactive-belief operator. Eigenvalues (red dots), resolvent-norm contours, dashed unit circle, dotted spectral-radius circle. In the symmetric population (left) eigenvalues and contours nearly coincide and $\Gamma\approx1.03$. In the feed-forward hierarchy (right) the eigenvalues cluster near the origin ($\spr=0.10$) while the resolvent norm bulges to $\lvert\zeta\rvert\approx0.95$---the signature of the non-normality that drives transient amplification.}
\label{fig:pseudospectra}
\end{figure}

This protrusion produces transient amplification. We report its dependence on hierarchy depth as a numerical finding rather than a theorem: the scaling is clean over the range we can compute, but---as we are careful to note---its asymptotic form is not established analytically, and in one tractable limit the gain saturates rather than growing. We state what the computations support and what they do not.

We first record what \emph{is} a theorem: the operator is non-normal and reactive at all, which is what makes the spectral certificate insufficient.

\begin{proposition}[Hierarchical reactivity is possible at arbitrarily small spectral radius]\label{prop:reactive}
For the feed-forward epistemic hierarchy with fast perception, the stacked operator $\Top$ is non-normal, and its transient gain strictly exceeds unity, $\Gamma=\sup_{k\ge0}\lVert\Top^k\rVert_2>1$, while its spectral radius $\spr(\Top)$ can be made arbitrarily small. Consequently a perturbation can be transiently amplified even though every eigenvalue lies arbitrarily deep inside the unit disc, so the spectral gap does not bound the transient response.
\end{proposition}

\begin{proof}
A strictly feed-forward influence matrix is nilpotent, and the self-weights contribute only a diagonal of magnitude $d$, so $\spr(\Top)$ is controlled by $d$ and tends to $0$ as $d\to0$. Non-normality is immediate: $\Top\Top^\top\neq\Top^\top\Top$ because the feed-forward coupling makes the belief--expression--perception loop directional (the superdiagonal transfer has no symmetric counterpart). For a non-normal contraction the one-step norm already exceeds the spectral radius, $\lVert\Top\rVert_2>\spr(\Top)$; taking $d$ small makes $\spr(\Top)$ arbitrarily small while $\lVert\Top\rVert_2$ remains bounded below by the off-diagonal transfer weight $w$, so $\Gamma\ge\lVert\Top\rVert_2>1$ with the gap between $\Gamma$ and $\spr(\Top)$ arbitrarily large. The Kreiss matrix theorem $\Kreiss(\Top)\le\Gamma\le e\,n\,\Kreiss(\Top)$ then certifies the pseudospectral (non-normal) origin of the amplification.
\end{proof}

The \emph{magnitude} of this reactivity, and its growth with hierarchy depth, we report as a numerical observation, with an explicit statement of scope. This is deliberately not labelled a theorem: while the data are clean over the computationally accessible range, we have not established the asymptotic scaling analytically, and---importantly---we find that it does not hold in every limit.

\begin{observation}[Growth of the transient gain with depth; numerical]\label{obs:sqrtN}
Over the accessible range $N\in\{3,\dots,28\}$, with perception rate $\gamma=0.9$ and spectral radius fixed at $\spr(\Top)=0.10$, the transient gain of the feed-forward hierarchy is well fit by
\begin{equation}
\Gamma(N)\approx 1+c\sqrt{N},\qquad c\approx1.0,\qquad R^2=0.9998,
\label{eq:sqrtN}
\end{equation}
a fit that dominates logarithmic ($R^2=0.98$) and linear ($R^2=0.975$) alternatives, and the Kreiss ratio $\Gamma/\Kreiss(\Top)\approx1.65$ is approximately constant across this range (Figure~\ref{fig:scaling}). Two caveats bound the claim. First, the normalized excess $(\Gamma-1)/\sqrt N$ has not stabilized over $N\le28$, so the coefficient $c$ and the precise exponent are empirical, not asymptotic. Second, in the instantaneous-perception limit $\gamma\to1$, where the system reduces to a $2N$-dimensional operator on $(\bvec,\xvec)$ that we can compute to $N=256$, the gain \emph{saturates} at $\Gamma\approx1.37$ rather than growing---so the depth-scaling~\eqref{eq:sqrtN} is a finite-$\gamma$ phenomenon and must not be read as an unbounded asymptotic law. What is robust across all our computations, and is the point that matters for the security argument, is the qualitative conclusion of Proposition~\ref{prop:reactive}: hierarchical structure makes $\Gamma$ substantially exceed $1$ at negligible spectral radius, and $\Gamma$ increases with depth over the range relevant to real institutions.
\end{observation}

The reactivity is thus a genuine and provable phenomenon (Proposition~\ref{prop:reactive}); its precise scaling with depth is an empirical regularity whose asymptotic status we leave open (Observation~\ref{obs:sqrtN}). The constancy of the Kreiss ratio over the studied range confirms that the amplification is of non-normal (pseudospectral) origin rather than a spectral effect in disguise. We flag that the mechanism is the coupled three-layer loop and not the influence chain alone: the raw feed-forward matrix $\Wmat$ has transient gain exactly $1$ (it is a nilpotent contraction with orthogonal-enough structure), and the reactivity appears only once belief, expression, and perception are coupled through~\eqref{eq:stacked}---a point we return to below.

\begin{figure}[t]
\centering
\includegraphics[width=0.62\linewidth]{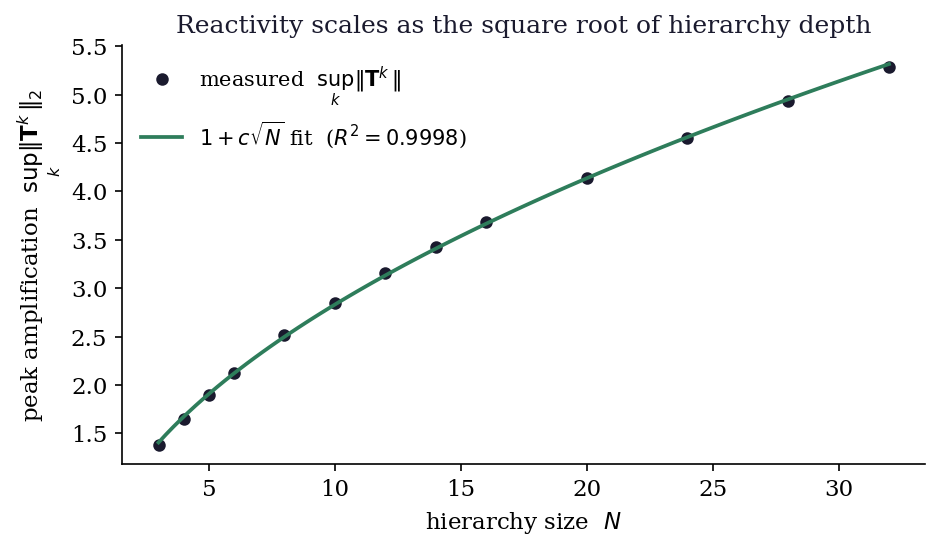}
\caption{Transient gain $\Gamma$ versus hierarchy depth $N$ over the accessible range, with the empirical fit $1+c\sqrt N$ ($R^2=0.9998$; Observation~\ref{obs:sqrtN}), at fixed spectral radius $\spr(\Top)=0.10$. The gain grows with depth over this range---a non-normal effect invisible to eigenvalue analysis---though, as noted in the text, the asymptotic form is not established and the gain saturates in the $\gamma\to1$ limit.}
\label{fig:scaling}
\end{figure}

\subsection{Adversarial exploitation of reactivity}

Reactivity is not a curiosity of the operator norm; it is an attack surface, and it makes Theorem~\ref{thm:transience} locally false as a security guarantee. The spectral certificate says a one-shot injection decays as $\spr(\Top)^k$. But an adversary who computes the leading right singular vector of $\Top^{k^*}$ at the gain-maximizing horizon $k^*$ obtains an injection direction that loads the perception and expression layers while placing almost nothing directly in the belief layer---and this injection transiently drives belief distortion far beyond what the spectrum permits. In a depth-$16$ hierarchy, the optimal injection produces a peak belief response $4.7$ times that of a random injection of equal magnitude, at a horizon $k^*$ where the spectral prediction says the perturbation should already have decayed (Figure~\ref{fig:reacttraj}). Because the effect recurs at the reactivity timescale, a patient adversary who repeats nominally ``transient'' attacks at intervals of $k^*$ sustains a distortion that the spectral theory declares impossible.

\begin{figure}[t]
\centering
\includegraphics[width=0.62\linewidth]{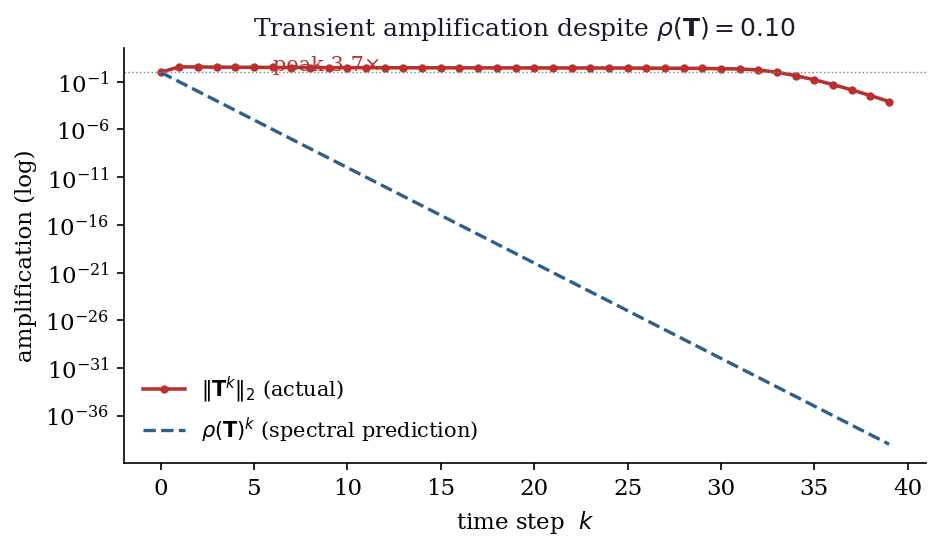}
\caption{Adversarial reactivity. Amplification trajectory $\lVert\Top^k\rVert$ (red) in a depth-$16$ hierarchy with $\spr(\Top)=0.10$, versus the spectral prediction $\spr^k$ (blue). A one-shot injection aligned to the transient drives belief distortion the spectrum forbids; the effect recurs at the reactivity horizon $k^*$, so repetition sustains it.}
\label{fig:reacttraj}
\end{figure}

The design implication reframes the robustness of social knowledge. Spectral gap certifies asymptotic safety but not transient safety, and the two come apart precisely for the directed, hierarchical structures that describe real influence: media pyramids, command chains, follower networks, citation hierarchies. Such structures may be spectrally quiescent---settling reliably, forgetting perturbations in the long run---and yet store transient attack capacity that scales with their depth, waiting to be pumped by an adversary who understands the operator's non-normality. Flat epistemic architectures are, in this precise sense, more robust than hierarchical ones of equal spectral radius: not because they equilibrate to different beliefs, but because they cannot be transiently pumped. To defend social knowledge is, in part, to flatten the hierarchies through which it is transmitted, or to monitor the deep ones at the reactivity timescale rather than trusting the spectral certificate. More broadly, the reactivity result places the security of belief networks within the operator-theoretic framework of non-normal dynamics, where the governing invariants are the pseudospectrum and the Kreiss constant rather than the spectrum, and where robustness is a property of the entire resolvent, not merely of its poles.

\part{Concluding Comments}\label{part:synthesis}

\section{What the theory says about social knowledge}\label{sec:implications}

We set out to give collective belief a formal object of study located between the individual and the aggregate---the structure of mutual attribution---and to derive from it a unified account of how social knowledge is formed, how it goes wrong, and how it can be defended. It is worth drawing the threads together, because the parts of the theory illuminate one another, and their conjunction says something about social epistemology that no part says alone.

The first lesson is that \emph{the unit of collective epistemology is the attribution, not the belief}. What moves through a group, what its members act on, what stabilizes into consensus or curdles into pluralistic ignorance, is not the distribution of private convictions but the tensor of who takes whom to believe what. The private belief is a special, diagonal case of the attribution; the interactive structure is primary. This inverts the usual order of explanation, on which group phenomena are to be built up from individual beliefs, and it is what allows the framework to represent states---false consensus that no one holds, corroboration with no source, expressed climates detached from private conviction---that a belief-first ontology cannot even name. The plural-subject and collective-intentionality traditions \citep{gilbert1989,tuomela2013,bratman2014} were right to insist that group belief is not the summation of individual belief; the present theory says what it is instead, and gives the constitution a mechanism.

The second lesson is that \emph{collective self-deception has an exact anatomy}. Misperception of others decomposes, cleanly and provably, into a perception-lag component that observation dissolves and a concealment component that observation cannot touch (Corollary~\ref{cor:channels}), and pluralistic ignorance is precisely the residue of the second after the first has closed (Theorem~\ref{thm:closure}). This yields the framework's sharpest and most counterintuitive claim: conformity \emph{amplifies but cannot generate} collective distortion. A population of sincere agents, however conformist, converges to no self-deception; distortion requires a bias in expression, a systematic reason to say other than one believes, and conformity then magnifies that seed without bound. The policy corollary is that the treatment for pluralistic ignorance is never more information and never more visibility---those close the channel that was already closing, and can feed the one that sustains the distortion (Corollary~\ref{cor:transparency})---but always a reduction in the cost of sincere dissent. The secret ballot, protected disagreement, and anonymity are not merely liberal niceties; they are, in this theory, the \emph{only} interventions that touch the equilibrium.

The third lesson is that \emph{the forms of collective epistemic life are continuously connected}. Consensus, polarization, entrenched distortion, and perpetual instability are not four sociologies but four regions of one parameter space (Proposition~\ref{prop:limits}), reached by turning conformity, the sign structure of trust, and the coherence of influence. A group can be moved among them by circumstance or by design, and small changes in the parameters of social-epistemic life---a rise in conformity, a souring of trust into antagonism, a strengthening of inferential links between issues---can carry a population across a qualitative boundary. This is at once a unification and a warning: the same machinery that produces healthy agreement produces, at nearby parameter settings, self-deception and division.

The fourth lesson is that \emph{manipulation is layered, and its remedies are not interchangeable}. Attacks on social knowledge are individuated by the epistemic layer they target---content, climate, or medium---and the layers differ in the durability of the distortions they admit and the remedies that repair them (Part~\ref{part:adversarial}). Attacks on states heal; attacks on the medium sediment into architecture and are impervious to disclosure; attacks on the expression environment produce ambient equilibria that transparency entrenches. The efficient attack shifts from content to climate as a population's conformity rises, so that the most conformist institutions are the most exposed to exactly the second-order manipulation their instincts are least equipped to counter (Corollary~\ref{cor:orgvuln}). And the single most common defensive error---answering every distortion with exposure---is correct for only one of the three families and counterproductive for another. A theory of the defense of social knowledge must begin by diagnosing the layer under attack.

The fifth lesson, from Part~\ref{part:physics}, is that \emph{spectral safety is not transient safety}. Hierarchical epistemic structures, which describe most real institutions of knowledge transmission, are non-normal and reactive (Proposition~\ref{prop:reactive}): they store transient attack capacity, invisible to the eigenvalue analysis that would certify them safe, and this capacity grows with hierarchy depth over the range we can compute (Observation~\ref{obs:sqrtN}). The robustness of social knowledge is a property of the whole resolvent, not of its poles, and the deep hierarchies through which knowledge usually flows are precisely the structures in which the gap between spectral and transient safety is largest.

\section{Limits of the linear theory}\label{sec:limits}

The results of this paper are theorems about a linear model, and intellectual honesty requires marking what that buys and what it costs. The linearity is not incidental; it is what makes the entire spectral, resolvent, and pseudospectral apparatus available, and it is what lets us prove rather than simulate. But it has three consequences that bound the theory's scope, each of which we have flagged locally and collect here.

First, \emph{superposition is why one-shot attacks wash out}. In a linear system perturbations decay independently of the state, so Theorem~\ref{thm:transience} holds globally. In a nonlinear system with thresholds, bounded beliefs, or cascade dynamics, a sufficiently large transient---amplified, perhaps, by the very reactivity of Part~\ref{part:physics}---can tip the system into a different basin of attraction, and the transience guarantee would hold only locally. The reactivity result should be read as identifying \emph{where} such tipping is most easily induced: the high-gain directions of a hierarchical operator are exactly where a bounded-belief nonlinearity would first be driven across a separatrix. A nonlinear successor theory would convert the transient excursions we prove into finite-amplitude regime changes, and the linear reactivity law would set the threshold.

Second, \emph{the divergences we prove are resolvent poles, not critical points}. The unbounded amplification of pluralistic ignorance as $\alpha\to1$ (Theorem~\ref{thm:closure}(d)) and of the echo multiplier as $\spr(\Bten)\to1$ (Theorem~\ref{thm:echo}) are poles of resolvents $(\Imat-\Kop)^{-1}$ and $(\Imat-\Bten)^{-1}$: they carry no diverging correlation length, no order parameter, no scaling exponent in the sense of statistical physics. They are the linear model's way of marking where its own idealization breaks down---where beliefs, if bounded, would saturate rather than diverge, and where nonlinearity must take over. We have been careful throughout to describe them as poles and not as phase transitions, and the reader should not read critical phenomena into them. (The reactivity of Part~\ref{part:physics} is a different kind of result---a property of the pseudospectrum rather than a resolvent pole. We are correspondingly careful there to separate what is proved, namely that hierarchical operators are non-normal and reactive at arbitrarily small spectral radius (Proposition~\ref{prop:reactive}), from what is only observed numerically, namely the approximate $\sqrt N$ growth of the transient gain over the accessible range, whose asymptotic form we explicitly leave open (Observation~\ref{obs:sqrtN}).)

Third, \emph{the theory prices attacks but does not model the attacker}. We have computed the cost, reach, and durability of manipulations, and read off the defenses, but we have treated the adversary as exogenous---a chooser of parameters and injections, not a strategic agent with objectives, beliefs, and a budget who anticipates the defender's response. The natural successor is a game: a defending designer and an optimizing adversary contesting the parameters of the epistemic network---conformity, the sign structure of trust, the honesty of the medium, the spectral slack, the Kreiss constant---each anticipating the other. The results of Part~\ref{part:adversarial} would then become the payoff structure of that game, and the defenses of Section~\ref{sec:defense} its equilibrium strategies. We regard this as the most important direction the framework opens.

Fourth, and most consequentially for a social epistemology, \emph{the theory is a dynamics of belief given a fixed epistemic-institutional scaffold, and it brackets the formation of that scaffold}. Three primitives are taken as exogenous: the proposition set $Q$ over which beliefs are held; the influence matrix $\Wmat$, whose sign structure and weights encode who attends to and trusts whom; and the inference operator $\Lop$ of Part~\ref{part:coupling}, which encodes the asserted entailments among propositions. In the framework these are inputs, and the theory says how belief moves through them. But where they come from is itself a social-epistemic question of the first importance, and it is the question the computational and institutional sociology of knowledge \citep[in the spirit of][]{evans2016} places at its center. That literature asks about the division of cognitive labor that determines who can know what, the organization of inquiry that fixes which propositions are even available to be believed, the path-dependence of attention and citation that shapes $\Wmat$, and the authority relations that decide whose entailments enter $\Lop$. A population's epistemic fate in our theory is largely set by these primitives: the spectral radius of $\Wmat$, the sign structure that decides consensus versus polarization, the conformity that governs pluralistic ignorance, the coupling in $\Lop$ that can destabilize. Yet the theory is silent on their genesis. We regard this not as a hidden defect but as a clean division of labor. The present paper contributes the dynamics on a given scaffold, exactly and tractably. The endogenous formation of the scaffold---how $Q$, $\Wmat$, and $\Lop$ are produced and reproduced by the institutions of knowledge---is a complementary inquiry. The natural point of contact between the two is the co-evolution of belief dynamics with the epistemic structure through which they run. An institutionally structured, endogenously evolving $\Lop$, in which asserted entailments are generated and contested by the organization of inquiry rather than fixed in advance, is the most direct bridge we see between the operator theory developed here and the content-and-institution focus of that literature.

Finally, the cardinal, real-valued representation of degrees of belief, and Assumption~\ref{ass:composition}'s projection of the interactive hierarchy onto composed point-attributions, are idealizations whose boundaries we have marked where they do work (Sections~\ref{sec:tensor} and~\ref{sec:higher}). The framework is offered not as the final word on the structure of interactive belief but as a demonstration that the structure is tractable---that the recursion of social cognition, so often treated as a source of irreducible complexity, yields to explicit analysis once the right object, the attribution tensor, is made the unit of study.

\section{Conclusion}\label{sec:conclusion}

What a group knows is not what its members know, nor the sum or average of what they know. It is a configuration of mutual attribution: of what each takes each to believe, express, and perceive. Its characteristic conditions, from shared understanding to collective self-deception, are configurations of that structure. We have argued that this configuration is the right unit of social-epistemic analysis, and we have represented it as a tensor. On that representation the life of collective belief becomes a tractable dynamical system. Its stable forms are spectral regimes of one operator. Its collective self-deceptions are the residue of a concealment channel that conformity amplifies but cannot create. Its higher-order structure is walks in a network. Its manipulations are layered attacks with distinct costs, durabilities, and remedies. Its deepest fragility is a transient, non-normal amplification that hierarchical structure stores and eigenvalue analysis cannot see. Each of these is a theorem in a system small enough to compute by hand and rich enough to represent the phenomena the study of social knowledge is concerned with: bubbles, echo chambers, pluralistic ignorance, influence operations. The framework does not make collective belief simple, or manipulation impossible. It makes the structure of the one, and the costs and remedies of the other, explicit and calculable. That is where a rigorous social epistemology of interactive belief can begin.

\appendix
\section{Proofs of the classical regime results}\label{app:proofs}

For completeness we prove the four regime results (Propositions~\ref{prop:consensus}--\ref{prop:oscillation}) that we restate from the classical literature within the interactive-belief formalism. These are not original; they are included so that the paper is self-contained and so that the objects on which the closure theorem (Theorem~\ref{thm:closure}) operates are fully specified.

\begin{proof}[Proof of Proposition~\ref{prop:consensus} (Consensus)]
By Perron--Frobenius, a primitive row-stochastic $\Wmat$ has a simple eigenvalue $1$ with right eigenvector $\one$ and positive left eigenvector $\pi$ (normalized to $\pi^\top\one=1$), all other eigenvalues strictly inside the unit disc. Let $\mathbf P=\one\pi^\top$; then $\mathbf P^2=\mathbf P$ and $\Wmat\mathbf P=\mathbf P\Wmat=\mathbf P$. Set $\mathbf R=\Wmat-\mathbf P$, so $\mathbf P\mathbf R=\mathbf R\mathbf P=\mathbf 0$ and $\Wmat^t=\mathbf P+\mathbf R^t$. The spectrum of $\mathbf R$ is that of $\Wmat$ with the Perron eigenvalue replaced by $0$, so $\spr(\mathbf R)<1$ and $\mathbf R^t\to\mathbf 0$. Hence $\bvec^t=\Wmat^t\bvec^0\to\mathbf P\bvec^0=\one\pi^\top\bvec^0=(\pi^\top\bvec^0)\one$.
\end{proof}

\begin{proof}[Proof of Proposition~\ref{prop:polarization} (Polarization)]
Since $\Dmat^2=\Imat$, we have $\Wmat^t=(\Dmat\Wmat_+\Dmat)^t=\Dmat\Wmat_+^t\Dmat$. By Proposition~\ref{prop:consensus}, $\Wmat_+^t\to\one\pi^\top$, so $\Wmat^t\to\Dmat\one\pi^\top\Dmat$ and $\bvec^t\to\Dmat\one(\pi^\top\Dmat\bvec^0)$. The vector $\Dmat\one$ has entries $\pm1$ according to camp, so the limiting beliefs are $\pm(\pi^\top\Dmat\bvec^0)$: equal magnitude, opposite sign by camp.
\end{proof}

\begin{proof}[Proof of Proposition~\ref{prop:distortion} (Stationary distortion)]
Since $\spr(\Wmat)<1$, $1$ is not an eigenvalue, so $\Imat-\Wmat$ is invertible and the Neumann series $\sum_{m\ge0}\Wmat^m$ converges to $(\Imat-\Wmat)^{-1}$. For constant forcing $\hvec^*$ the fixed point solves $\bvec^*=\Wmat\bvec^*+\hvec^*$, i.e.\ $\bvec^*=(\Imat-\Wmat)^{-1}\hvec^*$. Writing $\mathbf u^t=\bvec^t-\bvec^*$ gives $\mathbf u^{t+1}=\Wmat\mathbf u^t$, so $\mathbf u^t=\Wmat^t\mathbf u^0\to\mathbf 0$; a forcing converging to $\hvec^*$ adds a vanishing perturbation that does not affect the limit.
\end{proof}

\begin{proof}[Proof of Proposition~\ref{prop:oscillation} (Oscillation)]
Let $\lambda=e^{i\theta}$ with $\theta\neq0$ be an eigenvalue of unit modulus, with eigenvector $\mathbf w$. Decompose $\bvec^0$ over a basis including $\mathbf w$; the component along $\mathbf w$ evolves as $\lambda^t$. Since $\lvert\lambda^t\rvert=1$ the component neither grows nor decays, and since $\lambda^t=e^{i\theta t}$ with $\theta\neq0$ its phase rotates without limit, so $\lambda^t$ has no limit as $t\to\infty$. For any initial condition with nonzero projection onto $\mathbf w$---a generic condition---$\bvec^t$ therefore does not converge. (If $\lambda$ lies in a nontrivial Jordan block the component grows polynomially and again fails to converge.)
\end{proof}

\end{document}